\documentclass[aps,pra,showpacs,notitlepage,onecolumn,superscriptaddress,nofootinbib]{revtex4-1}
\usepackage[utf8]{inputenc}
\usepackage[tmargin=1in, bmargin=1.25in, lmargin=1.5in, rmargin=1.5in]{geometry}
\usepackage{amsmath, amssymb, amsthm}
\usepackage{graphicx}
\usepackage{xcolor}
\usepackage{enumitem}
\usepackage{datetime}
\usepackage{hyperref}
\usepackage{titlesec}
\usepackage{import}
\usepackage{mathtools}
\usepackage{thmtools,thm-restate}
\usepackage{comment}

\definecolor{crimson}{RGB}{186,0,44}
\definecolor{moss}{RGB}{0, 186, 111}

\hypersetup{
    colorlinks,
    linkcolor={crimson},
    citecolor={crimson},
    urlcolor={crimson}
}

\usepackage{qcircuit}

\theoremstyle{definition}
\newtheorem{definition}{Definition}[section]
\newtheorem{lemma}{Lemma}[section]
\newtheorem{theorem}{Theorem}[section]

\newtheorem*{theorem*}{Theorem}
\newtheorem*{corollary*}{Corollary}
\newtheorem{remark}{Remark}[section]

\usepackage{geometry}
\begin{document}

\title{Non-Markov quantum belief propagation}

\author{Jack Ceroni}
\email{jc@extropic.ai}
\affiliation{Extropic Corp., San Francisco, California, 94111, USA}
\affiliation{Department of Mathematics, University of Toronto, Toronto, Ontario M5S1A1, Canada}

\author{Ian MacCormack}
\affiliation{Extropic Corp., San Francisco, California, 94111, USA}

\author{Guillaume Verdon}
\affiliation{Extropic Corp., San Francisco, California, 94111, USA}

\date{\today}

\begin{abstract}
\noindent We provide a rigorous proof of the approximate convergence of sliding-window quantum belief-propagation as outlined heuristically in the work of Bilgin and Poulin (Ref.~\cite{bilgin2010coarse}), in the absence of the quantum Markov property. In particular, we confirm the hypothesis outlined in this work that the approximation error of each step in the belief-propagation algorithm decreases exponentially with the sliding-window size, under the assumption that the underlying state on which belief-propagation is being performed possesses a so-called thermal boundedness property: a relaxation of the Markov property required for exact convergence.
\end{abstract}

\maketitle
\section{Introduction}
\noindent The problem of simulating large-scale quantum systems is one of the pinnacles of modern computational physics, with applications to chemistry, material science, electronics engineering, drug design, and more \cite{georgescu2014quantum}. Although most problems under the umbrella of quantum simulation are computationally difficult, there have been numerous attempts over the past several decades to find restricted problems that are provably efficiently solvable, as well as heuristics which have good empirical performance, but lack guarantees about robustness. These methods include tensor network simulations \cite{2006quant.ph..8197P, 2019NatRP...1..538O}, Hartree-Fock and density functional methods in quantum chemistry \cite{casares2023grad, sun2018pyscf, hohenberg1964inhomogeneous}, and machine-learning based approaches \cite{2016arXiv160102036A,2022PhRvA.106a2409K}. In addition to the proliferation of this classical methods, one of the main proposed use-cases of quantum computers is simulation and characterization of large-scale quantum systems, due to the fact that a quantum device can store and manipulate quantum states natively. This has led to proposals for performing quantum simulation with fault-tolerant quantum computers running algorithms which likely yield speedups over classical methods (when executed on fault-tolerant quantum computers) \cite{1995quant.ph.11026K}, as well as variational methods meant to run on nearer-term devices \cite{2020arXiv201209265C}, but which lack most of the complexity-theoretic guarantees on efficient scaling which are true of fault-tolerant algorithms.

Among the proposals to simulate quantum systems with quantum devices, the two problems which have received the most attention are time-evolution and thermal state calculations. Time-evolution refers to computing the dynamics of a quantum system as it evolves under the influence of a particular Hamiltonian, while thermal state calculations refer to preparation and measurement of thermal states: the state that a quantum system with a fixed Hamiltonian achieves when it reaches thermal equilibrium at a particular temperature. The problem addressed in this paper is the latter: measurement of thermal states. In particular, we are interested in computing physically observable quantities associated with systems which are in their thermal state. Oftentimes, such quantities can be calculated with only \textit{local knowledge} of each constituent ``component" of the thermal state. In particular, for a quantum system distributed across many qubits (or other degrees of freedom), the local state of the system at a given sub-collection of qubits is referred to as the \textit{reduced density operator}. For a given thermal state, it is possible to compute its corresponding reduced density operators by preparing the whole state (an operator which scales exponentially in dimension as we increase the number of qubits), and then take partial traces onto subsystems of interest. This naive approach, however, exhibits poor scalability due to the exponential growth of the operator's size. In this paper, we explore a method first proposed by Liefer and Poulin in Ref.~\cite{leifer2008quantum} know as \textit{quantum belief-propagation via quantum message-passing}. In analogy to classical belief propagation \cite{10.5555/779343.779352}, quantum belief propagation, under certain (restrictive) conditions, is able to prepare reduced density operators of a thermal state \emph{without} first preparing the global thermal state and performing a partial trace. Moreover, the class of density operators for which this method works exactly is \emph{fundamentally distinct} from those which can be approximated via tensor network methods (i.e. long range entanglement as realized by bond dimensions is not the limiting factor) \cite{2006PhRvB..73i4423V}. However, like tensor network methods, the restrictions on the class of states for which QBP is guaranteed to give the correct answer are often difficult to fulfill. Of the required constraints, the most restrictive is the \emph{Markov property}. Ref.~\cite{leifer2008quantum} makes note of this fact, and suggests a technique for circumventing this constraint, via a protocol called \emph{sliding-window belief propagation}, where the underlying thermal state must only satisfy a weaker version of the Markov property, which we refer to as the $\ell$-Markov property, at the cost of increased resource requirements. Extending upon this work, in Ref.~\cite{bilgin2010coarse}, Bilgin and Poulin treat sliding-window belief-propagation as a heuristic, parameterized by the size of the so-called ``sliding window" used. In this prior work, it is demonstrated numerically, and via some physical intuition, that certain quantum graphical models will approximately satisfy a constraint implied by the $\ell$-Markov property which is sufficient to guarantee approximate convergence of sliding-window belief propagation for a size-$\ell$ window.

Extending upon the promising work of Ref.~\cite{bilgin2010coarse}, the goal of this paper is to take a step towards putting these results on a firmer theoretical footing. To this end, we derive rigorous error bounds for a single step of the sliding-window belief-propagation algorithm, under the assumption that the underlying state satisfies a property alluded to in Ref.~\cite{bilgin2010coarse} which we term \emph{thermal boundedness}. Our hope is that this analysis is a first step towards better understanding of the feasibility of using the sliding-window belief-propagation algorithm for the practical application of large-scale quantum simulation.

This manuscript is structured as follows. In Sec.~\ref{sec:background}, we briefly introduce the background information related to quantum belief-propagation via message-passing, and exposit its generalized version, sliding-window belief-propagation, as well as the property of thermal boundedness, formalized from ideas in Ref.~\cite{bilgin2010coarse}. In Sec.~\ref{sec:main_results}, we outline the main contribution of this work: the theorem which characterizes the error introduced in each step of sliding-window belief propagation. We conclude the paper by summarizing the results and their implications, and suggest directions of future inquiry in Sec.~\ref{sec:conclusion}. In Appx.~\ref{appx:proofs}, we provide a proof of the main theorem, Thm.~\ref{thm:error_single_step}, as well as supporting lemmas. Appx.~\ref{appx:markov} presents a supplementary information related to quantum belief-propagtion and sliding-window algorithms, as well as further discussion about computation of the circle product (Def.~\ref{def:circle_prod}). Finally, Appx.~\ref{appx:misc} is a collection of miscellaneous results which are utilized throughout the paper.

\section{Background}
\label{sec:background}

\noindent We begin with a discussion of context and briefly provide relevant background information related to quantum belief-propagation via message-passing, and its generalization, sliding-window belief-propagation.

\subsection{Quantum belief-propagation via message-passing}

\noindent The procedure of quantum belief-propagation via message-passing allows for the computation of reduced density matrices of state $\rho_\mathcal{V}$, in some \textit{quantum graphical model} $(G, \rho_V)$ satisfying a set of assumptions, in a computationally efficient manner. As an aside, we note that all of the definitions below have direct equivalents in terms of quantum factor graphical models, a formalism that may sometimes be more convenient \cite{leifer2008quantum}.
\begin{definition}[Quantum Graphical Model]
  Let $G = (\mathcal{V}, \mathcal{E})$ be a graph. Let $\mathcal{H} = \bigotimes_{v \in \mathcal{V}} \mathcal{H}^{v}$ be a quantum system supported on the nodes of $G$. Let us have a density operator $\rho_{\mathcal{V}} \in \bigotimes_{v \in V} \mathcal{H}^{v}$.
  Then the pair $(G, \rho_{\mathcal{V}})$ is said to be a quantum graphical model.
\end{definition}
\noindent In the case that $G$ is a tree graph (there exists a unique path between any two nodes of $G$), and $\rho_\mathcal{V}$ is a \textit{bifactor state of $G$} with the \textit{Markov property}, then the quantum message-passing algorithm converges.

Let us now define a useful operation which will be utilized throughout the paper:

\begin{definition}[Circle Product]
\label{def:circle_prod}
  Given positive Hermitian operators $A$ and $B$, we define the circle product between $A$ and $B$ as
  \begin{equation}
    A \odot B = \exp \left( \log(A) + \log(B) \right)
  \end{equation}
  where $\exp$ is the matrix exponential, and $\log$ is the matrix logarithm.
\end{definition}
\noindent The intuition for this definition is that it combines the Hamiltonians of two thermal states. In particular, $e^{A} \odot e^{B} = e^{A + B}$. Note that when $[A, B] = 0$, the circle product reduces to the standard matrix product: $e^{A} \odot e^{B} = e^{A} e^{B}$.
\begin{definition}[Bifactor State]
\label{def:bifactor}
  A state $\rho_\mathcal{V}$ defined on a graph $G = (\mathcal{V}, \mathcal{E})$, with $\rho_\mathcal{V} \in \bigotimes_{v \in \mathcal{V}} \mathcal{H}^{v}$ is said to be a bifactor state of $G$ if we can write
  \begin{equation}
   \rho_\mathcal{V} = \frac{1}{\mathcal{Z}} \left[ \bigodot_{v \in \mathcal{V}} \mu_v \right] \odot \left[ \bigodot_{e \in \mathcal{E}} \nu_{e} \right] 
  \end{equation}
  where $\mu_v$ is supported on $\mathcal{H}^{v}$, and $\nu_e$ for $e = (e_1, e_2)$ is supported on $\mathcal{H}^{e_1} \otimes \mathcal{H}^{e_2}$.
\end{definition}
\begin{definition}[Conditional Mutual Information]
  Given a density operator $\rho_{ABC} \in \mathcal{H}^{A} \otimes \mathcal{H}^{B} \otimes \mathcal{H}^{C}$, we say the conditional mutual information between systems $A$ and $C$, given $B$, is defined to be
  \begin{equation}
    S(A : C | B) = S(\rho_{AB}) + S(\rho_{BC}) - S(\rho_{ABC}) - S(\rho_B)
  \end{equation}
  where $S$ is the von Neumann entropy.
\end{definition}
\begin{definition}[Markov Property]
\label{def:q_markov}
  A quantum graphical model $(G, \rho_\mathcal{V})$ is said to have the Markov property if for any subset $\mathcal{U} \subset \mathcal{V}$ with $G = (\mathcal{V}, \mathcal{E})$, we have
  \begin{equation}
    S( \mathcal{U} : \mathcal{V} - (n(\mathcal{U}) \cup \mathcal{U}) | n(\mathcal{U}) - \mathcal{U}) = 0,
  \end{equation}
  where $n(\mathcal{U})$ denotes all nodes neighbouring a node of $\mathcal{U}$, in $G$. The set $(n(\mathcal{U}) \cup \mathcal{U})$ is sometimes referred to as the \textit{Markov blanket} of $\mathcal{U}$. So, intuitively, the Markov property is satisfied for a subset $\mathcal{U}$ if $\mathcal{U}$ is independent of $V-(n(\mathcal{U}) \cup \mathcal{U})$, conditioned on $(n(\mathcal{U}) \cup \mathcal{U})$, the Markov blanket of $\mathcal{U}$.
\end{definition}

\noindent As is implied by the name, quantum belief-propagation via message-passage effectuates the calculation of reduced density matrices by passing \emph{messages} along the edges of a quantum graphical model.
\begin{definition}[Messages and beliefs]
   Given a graphical model $(G, \rho_\mathcal{V})$ where $\rho_{\mathcal{V}}$ is a bifactor state as in Def.~\ref{def:bifactor}, the message passed from node $u$ to node $v$ in $\mathcal{V}$, at discrete time $t \geq 1$ is defined, inductively, to be the operator
    \begin{equation}
    \label{eq:message}
      m_{u \to v}(t) = \frac{1}{\mathcal{Z}} \text{Tr}_{u} \left[ \mu_u \odot \nu_{(u, v)} \odot \left( \bigodot_{v' \in n(u) - v} m_{v' \to u}(t - 1) \right) \right],
    \end{equation} 
    and $m_{u \to v}(0) =  \frac{1}{\mathcal{Z}} \text{Tr}_{u} \left[ \mu_u \odot \nu_{(u, v)} \right]$. In addition, the \textit{belief} (an estimate of the reduced density operator) at node $v$, at time $t$, is defined in terms of the messages as
    \begin{equation}
        b_v(t) = \frac{1}{\mathcal{Z}} \left[ \mu_{v} \odot \left( \bigodot_{u \in n(v)} m_{u \to v}(t) \right) \right].
    \end{equation}
\end{definition}
\noindent We can now present the main theorem characterizing message-passing.
\begin{theorem}[Quantum belief-propagation]
\label{thm:qmp}
Given a quantum graphical model $(G, \rho_\mathcal{V})$ and a node $v \in \mathcal{V}$, there exists some $T$ (in particular, when $T = \max_{w \in \mathcal{V}} d(v, w)$, the largest distance between $v$ and another node of the graph) such that $b_v(T) = \mathrm{Tr}_{\mathcal{V} - v}[\rho_\mathcal{V}]$: the reduced density matrix of $\rho_\mathcal{V}$ on node $v$ if

\begin{enumerate}
    \item $G$ is a tree graph (there is a unique path between each pair of nodes),
    \item $\rho_\mathcal{V}$ is a bifactor state (Definition~\ref{def:bifactor}),
    \item $(G, \rho_\mathcal{V})$ has the quantum Markov property (Definition~\ref{def:q_markov}).
\end{enumerate}
\end{theorem}
\noindent For further discussion of this result, see Appx.~\ref{appx:markov}. At a high-level, the belief-propagation algorithm relies on the fact that given a bifactor state $\rho_{\mathcal{V}}$, and assuming that this state has the Markov property, the superoperation of taking the partial trace over a particular site of $\rho_{\mathcal{V}}$ \emph{commutes with the circle product} (this is a consequence of Thm.~\ref{thm:decomp}). Therefore, if $G$ has tree-structure rooted at $v \in \mathcal{V}$, one is able to propagate partial traces through the network down to the leaf nodes, recursively shrinking the support of $\rho_{\mathcal{V}}$ until all that remains is the reduced state at site $v$. During this procedure, because of the local circle product structure, the only states which have to be prepared at a given step are those which are being traced over (which, as was stated, are localized to the leaf nodes). It is the breakdown of this commutativity of the circle product with partial traces which leads to difficulties when attempting to run belief-propagation on non-Markov states.

\subsection{Sliding-window belief propagation and thermal boundedness}

\noindent Thm.~\ref{thm:qmp} is a powerful result, seeing as it guarantees exact convergence to a desired reduced density matrix, in a fixed amount of time. However, as was mentioned previously, its restrictions can pose difficulties, the most notable being the Markov condition, which is often difficult to satisfy.

To further place this property's necessity in context, let us consider the pedagogical example where we are provided with a quantum bifactor network $(\rho_\mathcal{V}, G)$ such that $G$ is an $N$-vertex chain. We write $\rho_\mathcal{V}$ as a circle product of edge terms, that is
\begin{equation}
    \rho_{\mathcal{V}} = \frac{1}{\mathcal{Z}} \left[ e^{-\beta h_{N, N - 1}} \odot \cdots \odot e^{-\beta h_{3, 2}} \odot e^{-\beta h_{2, 1}} \right]
\end{equation}
where $h_{k, k - 1}$ is supported on $\mathcal{H}^{v_k} \otimes \mathcal{H}^{v_{k - 1}}$. Suppose we are interested in computing the reduced density matrix on node $v_N$, $\rho_{v_N}$. This is equivalent to calculation of the partial trace over all other nodes, which we can perform sequentially, $\rho_{v_N} = (\text{Tr}_{v_{N - 1}} \circ \cdots \circ \text{Tr}_{v_{1}})(\rho_\mathcal{V})$. If $\rho_\mathcal{V}$ were a Markov model, it would be possible to propagate each of the successive partial traces through the sequence of circle products due to the commutativity of circle products with partial traces, mentioned above. For example, taking the partial trace on the first site can be done via the identity
\begin{equation}
    \label{eq:markov_chain}
    \text{Tr}_1 \left[ e^{-\beta h_{N, N - 1}} \odot \cdots \odot e^{-\beta h_{3, 2}} \odot e^{-\beta h_{2, 1}} \right] = e^{-\beta h_{N, N - 1}} \odot \cdots \odot e^{-\beta h_{3, 2}} \odot \text{Tr}_1 \left[ e^{-\beta h_{2, 1}} \right].
\end{equation}
From Lem.~\ref{lem:mar_tr}, we know that taking the partial trace preserves the Markov property, so we can repeat this procedure for all other $\text{Tr}_{v_k}$ operations. This implies that the desired reduced density operator can be computed via a sequence of circle products and partial traces, where the largest thermal state ever being prepared during this protocol is of size $\max_k \dim(\mathcal{H}^{v_k} \otimes \mathcal{H}^{v_{k - 1}})$. In the case, however, that $\rho_{\mathcal{V}}$ \textit{does not} have the Markov property, Eq.~\ref{eq:markov_chain} does not hold, and the partial trace operations do not commute with the circle products.

Let us consider a relaxation of the Markov property.
\begin{definition}[$\ell$-Markov property]
   A quantum graphical model $(G, \rho_\mathcal{V})$ is said to have the $\ell$-Markov property if for any subset $\mathcal{U} \subset \mathcal{V}$ with $G = (\mathcal{V}, \mathcal{E})$, we have
  \begin{equation}
    S( \mathcal{U} : \mathcal{V} - (n_{\ell}(\mathcal{U}) \cup \mathcal{U}) | n_{\ell}(\mathcal{U}) - \mathcal{U}) = 0,
  \end{equation}
  where $n_{\ell}(\mathcal{U})$ denotes all nodes $v \in \mathcal{V}$ such that there exists some $u \in \mathcal{U}$ with $d(v, u) \leq \ell$. Note that in the case that $\ell = 1$, the above definition reduces to the original Markov property.
\end{definition}
\noindent Returning to the simple chain example, suppose $\rho_\mathcal{V}$ were to possess this relaxed property for some $\ell$. Then, we would have
\begin{equation}
    \label{eq:sw}
    \text{Tr}_1 \left[ e^{-\beta h_{N, N - 1}} \odot \cdots \odot e^{-\beta h_{2, 1}} \right] = e^{-\beta h_{N, N - 1}} \odot \cdots \odot e^{-\beta h_{\ell + 2, \ell + 1}} \odot \text{Tr}_1 \left[ e^{-\beta h_{\ell + 1, \ell}} \odot \cdots \odot e^{-\beta h_{2, 1}} \right].
\end{equation}
If we continue, taking the $k$-th partial trace before the $(\ell + k)$-th term of the circle product, then message-passing can carry forward in a similar fashion as the original protocol (as this propagation of partial traces is effectively what is happening during the process of message-passing), at the cost of having to prepare larger ``messages" (in this case, of size $\max_{k} \dim(\mathcal{H}^{v_k} \otimes \cdots \otimes \mathcal{H}^{v_{k + \ell}})$). 

In Ref.~\cite{bilgin2010coarse}, Bilgin and Poulin, inspired by this relaxation, demonstrate empirically that the partial trace and circle product approximately commute to a much higher degree when the partial trace is being taken over a site of a graphical model which lies far from the terms directly involved in the circle product (i.e. as $\ell$ becomes larger). Moreover, they observe this property without any consideration of the Markov property. They refer to this new method of belief-propagation as \emph{sliding-window belief propagation}. They justify this fact by reasoning and numerically demonstrating that the influence of the so-called \emph{thermal potential} induced on a given thermal state by the terms being traced over, decays rapidly as distance grows between the site being traced, and the site to which the partial trace is being propagated (i.e. in a chain, this would be the site at distance $\ell$ away from the site being traced over).

\begin{definition}[Thermal Potential]
\label{def:thermal_potential}
Let $G = (\mathcal{V}, \mathcal{E})$, let $H_ = \sum_{e \in \mathcal{E}} h_e$ be a local Hamiltonian supported on $\bigotimes_{v \in \mathcal{V}} \mathcal{H}^{v}$. Define the thermal potential of $H$ effectuated by $\mathcal{V}' \subset \mathcal{V}$ at temperature $1/\beta$ as
\begin{equation}
    V_{\text{thermal}}(H; \mathcal{V}', \beta) = -\frac{1}{\beta} \log \left( \text{Tr}_{\mathcal{V}'}\left[ \frac{1}{\mathcal{Z}} e^{-\beta H} \right] \right) - H
\end{equation}
where $\text{Tr}_{\mathcal{V}'}$ is the partial trace over sites contained in $\mathcal{V}'$.
\end{definition}
\noindent The intuition for this definition lies in the fact that it characterizes the change in the effective Hamiltonian of the thermal state before and after tracing out the nodes of $\mathcal{V}'$. In particular,
\begin{equation}
    \text{Tr}_{\mathcal{V}'}\left[ \frac{1}{\mathcal{Z}} e^{-\beta H} \right] = e^{-\beta \left[ H + V_{\text{thermal}}(H; \mathcal{V}', \beta) \right]}
\end{equation}
Note that, although the argument of the exponent on the right-hand side above contains the full Hamiltonian $H$, it is only nontrivially defined on $\mathcal{V}-\mathcal{V}'$, as one would expect as the result of a partial trace over $\mathcal{V}'$. Due to the fact that we expect local perturbations (i.e. taking the partial trace at a single site) to have local influence on the resulting thermal state in many systems of interest, it is sensible to conclude that $V_{\text{thermal}}(H, \mathcal{V}')$ will be \emph{short-ranged}. This is to say that non-trivial terms of this potential, which have support lying farther and farther away from $\mathcal{V}'$, become small rapidly. More formally, we say that the sequence of magnitudes of the \emph{cumulants away from $\mathcal{V}'$} of $V$ becomes small quickly.
\begin{definition}[Cumulants of a local operator]
\label{def:cum}
Given an operator $O$ supported on graph $G$, the cumulants of $O$ away from $\mathcal{V}'$ are defined inductively as
\begin{equation}
    O^{(j)} = \text{Tr}_{D_j(\mathcal{V}')} \left[ O - \displaystyle\sum_{k = 0}^{j - 1} O^{(k)} \right] \ \ \ \text{with} \ \ \ O^{(0)} = 0,
\end{equation}
where $D_j(\mathcal{V}') = \{v \ | \ d(v, \mathcal{V}') > j\}$. The sequence of cumulants is then $O^{(1)}, O^{(2)}, \dots$.
\end{definition}
\noindent This sequence of operators can be thought of as individual contributions to the total action of $O$ whose ``influence" lying closest to a given region occurs at distance $j$. In particular, note that $O = \sum_{j} O^{(j)}$, where
$O^{(1)}$ represents the contribution to $O$ supported at a distance of at most $1$ from $\mathcal{V}'$, $O^{(2)}$ represents the contribution to $O$ supported at a distance of at most $2$, \emph{that is not attributable to the contribution at distance $1$}, $O^{(3)}$ is the contribution at distance at most $3$ which is not attributable to the distance $1$ or $2$ contributions, and so on.

Ref.~\cite{bilgin2010coarse} demonstrates numerically that in the case of a transverse field, one-dimensional Ising model chain, the decay of cumulants is exponentially bounded in magnitude, and moreover, goes to zero much faster than correlation coefficients between sites separated by increasing integer lengths. Of course, this property may not hold generically, but it seems reasonable to expect that for systems with decaying correlations, it should too be the case that the cumulants of the thermal potential decay as well. This hypothesis leads us to make the following definition.

\begin{restatable}[Thermal boundedness]{definition}{thermalbnd}
    Given a graph $G = (\mathcal{V}, \mathcal{E})$ and $H = \sum_{e \in \mathcal{E}} h_e$, we say that $H$ is \emph{thermally bounded} if, for any $\mathcal{E}' \subset \mathcal{E}$, $H' = \sum_{e \in \mathcal{E}'} h_e$, and $\mathcal{V}' \subset \mathcal{V}$ supported on the endpoints of elements of $\mathcal{E}'$, there exists $K(\beta), k(\beta) \in \mathbb{R}^{+}$ (where we assume $K(\beta) \in O(\exp(\beta))$ and $k(\beta) \in O(\beta^{-1})$ in the worst case), such that
    \begin{equation}
        || V_{\text{thermal}}^{(j)}(H'; \mathcal{V}', \beta)|| \leq K(\beta) e^{-k(\beta) j}
    \end{equation}
    where $V_{\text{thermal}}^{(j)}(H'; \mathcal{V}', \beta)$ is the $j$-th cumulant away from $\mathcal{V}'$ at temperature $1/\beta$, given in Def.~\ref{def:cum}. We say that a thermal state is thermally bounded if its corresponding effective Hamiltonian is thermally bounded.
\end{restatable}

\noindent By assuming this property, we can prove rigorous bounds on the error between the left-hand and right-hand sides of Eq.~\eqref{eq:sw}, which involve the behaviour of the thermal potential, as well as other quantities characterizing the system.

\begin{remark}[Thermal boundedness of Markov states]
Before proceeding, let us briefly note how the above definition subsumes the original Markov property. Suppose $H$ is a graphical Hamiltonian on graph $G = (\mathcal{V}, \mathcal{E})$ of the form $H = \sum_{e \in \mathcal{E}} h_e$. Pick some $\mathcal{E}' \subset \mathcal{E}$ and let $H' = \sum_{e \in \mathcal{E}'} h_e$ so that the corresponding thermal state
\begin{equation}
    e^{-\beta H'} =  \bigodot_{e \in \mathcal{E}'} e^{-\beta h_e}
\end{equation}
has the Markov property, where we assume that $e^{-\beta H'}$ has unit trace without loss of generality. It follows from the Markov property, and thus commutativity of the partial trace with circle products, that for some $\mathcal{V}' \subset \mathcal{V}$, we will have
\begin{equation}
    \text{Tr}_{\mathcal{V}'}\left[e^{-\beta H'}\right] = \left[ \bigodot_{e \in \mathcal{E}'_\text{out}} e^{-\beta h_e} \right] \odot \text{Tr}_{\mathcal{V'}} \left[ \bigodot_{e \in \mathcal{E}'_\text{in}} e^{-\beta h_e} \right]
\end{equation}
where $\mathcal{E}'_{\text{out}}$ is the subset of $\mathcal{E}'$ having no endpoints in $\mathcal{V}'$, while $\mathcal{E}_{\text{in}}'$ is the subset of $\mathcal{E}'$ having at least one endpoint in $\mathcal{V}'$. From here, we have
\begin{align}
    V_{\text{thermal}}(H'; \mathcal{V}', \beta) &= \displaystyle\sum_{e \in \mathcal{E}'_{\text{out}}} h_e - \frac{1}{\beta} \log \left( \text{Tr}_{\mathcal{V}'} \left[ \bigodot_{e \in \mathcal{E}'_{\text{in}}} e^{-\beta h_e} \right] \right) - \displaystyle\sum_{e \in \mathcal{E}'} h_e
    \\ &= \displaystyle\sum_{e \in \mathcal{E}'_{\text{in}}} h_e - \frac{1}{\beta} \log \left( \text{Tr}_{\mathcal{V}'} \left[ \bigodot_{e \in \mathcal{E}'_{\text{in}}} e^{-\beta h_e} \right] \right).
\end{align}
This operator is supported entirely on $v \in \mathcal{V}$ which lie at distance at most $1$ from $\mathcal{V}'$. It follows that $V_{\text{thermal}}^{(j)}(H'; \mathcal{V}', \beta) = 0$ for $j \geq 2$. It follows immediately that $H$ is thermally bounded, provided that there is an upper bound on $V_{\text{thermal}}$ which scales at worst exponential in $\beta$ (a fact which is straightforward to check).
\end{remark}

\section{Main result}
\label{sec:main_results}

\noindent Having presented the relevant theory, we move on to presenting the main result of this work: characterizing the error introduced in a single step of the sliding-window belief propagation algorithm, for a given \emph{sliding-window size} (the parameter $\ell$).

First, some basic terminology. Given a tree graph $G = (\mathcal{V}, \mathcal{E})$, and a region $\mathcal{V}' \subset \mathcal{V}$, define 
\begin{equation}
B_{\ell}(\mathcal{V}') = \{(v, w) \ | \ (v, w) \in \mathcal{E}, d(v, \mathcal{V}') = \ell \ \text{or} \ d(w, \mathcal{V}') = \ell\}.
\end{equation}
where $d(v, \mathcal{V}')$ is the shortest distance along graph edges between $v$ and some vertex contained in $\mathcal{V}$. In addition, let 
\begin{align}
    L_{\ell}(\mathcal{V}') = \{(v, w) \ | \ (v, w) \in \mathcal{E}, d(v, \mathcal{V}'), d(w, \mathcal{V}') > \ell\}, \\ 
    R_{\ell}(\mathcal{V}') = \{(v, w) \ | \ (v, w) \in \mathcal{E}, d(v, \mathcal{V}'), d(w, \mathcal{V}') < \ell\}.
\end{align}

\begin{figure}
    \centering
    \includegraphics[width=300pt]{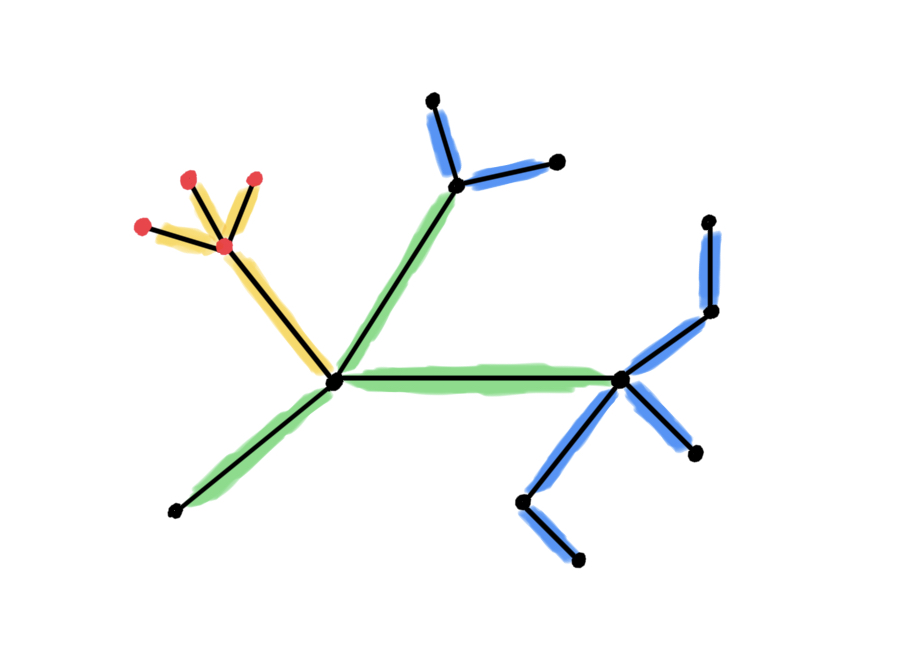}
    \caption{A diagram showing a tree graph, with a specified region $\mathcal{V}' \subset \mathcal{V}$ (the red/light-shaded nodes), along with the regions $B_{\ell}(\mathcal{V}')$ (the green/three central edges), $L_{\ell}(\mathcal{V}')$ (the blue/outer dark-shaded edges), and $R_{\ell}(\mathcal{V}')$ (the yellow edges, those adjacent to the lightly-shaded nodes), for distance $\ell = 0$.}
\end{figure}

\noindent Further, given a local Hamiltonian $H = \sum_{e \in \mathcal{E}} h_e$, define $H_{\mathcal{E}'}$ for $\mathcal{E}' \subset \mathcal{E}$ as $\sum_{e \in \mathcal{E}'} h_e$. We arrive at the following theorem.

\begin{restatable}[Error in single sliding-window belief propagation step]{theorem}{errorsinglestep}
\label{thm:error_single_step}
Let $(G, \rho_{\mathcal{V}})$ be a quantum bifactor network on tree graph $G = (\mathcal{V}, \mathcal{E})$, where $\rho_{\mathcal{V}}$ is the thermal state at temperature $1/\beta$ of the local Hamiltonian $H = \sum_{e \in \mathcal{E}} h_e$. Let $v^{*}$ be a leaf node of the tree. Assume that $H$ is thermally bounded with respect to functions $K(\beta)$ and $k(\beta)$. Then, there exist efficiently computable, non-negative $\widetilde{F}, \widetilde{G} \in \mathcal{O}(\exp(||H_B|| \beta))$ and efficiently computable, non-negative $\widetilde{f} \in \mathcal{O}(\beta^{-1})$ such that for some $\ell \in \mathbb{Z}^{+}$,
\begin{equation}
    \left|\left| \text{Tr}_{v^{*}}\left[ \frac{1}{\mathcal{Z}} e^{-\beta H} \right] - \frac{1}{\mathcal{Z}} \left[ e^{-\beta (H_L + H_B)} \odot \text{Tr}_{v^{*}}\left[ e^{-\beta H_R} \right] \right] \right|\right|_1 \leq (\ell \widetilde{F}(\beta, H_B) + \widetilde{G}(\beta, H_B)) e^{-\widetilde{f}(\beta) \ell}
\end{equation}
where $H_B \coloneqq H_{B_{\ell}(v^{*})}$, $H_L \coloneqq H_{L_{\ell}(v^{*})}$, and $H_R \coloneqq H_{R_{\ell}(v^{*})}$.
\end{restatable}

\noindent The proof of this lemma relies heavily on results of Hastings, which characterize perturbations of thermal states, \emph{within the exponential}, via conjugation of the unperturbed thermal state by a local operator (Ref.~\cite{hastings2007quantum}).

    \begin{restatable}[Hastings' belief-propagation]{theorem}{hastings}
    \label{thm:hastings}
    Let $H(s) = H + sV$ be a family of local Hamiltonians supported on $\bigotimes_{v \in \mathcal{V}} \mathcal{H}_v$ for graph $G = (\mathcal{V}, \mathcal{E})$, where $s \in [0, 1]$. Then
    \begin{equation}
        \label{eq:anti_commute}
        \frac{d}{ds} e^{-\beta H(s)} = -\frac{\beta}{2} \left\{ e^{-\beta H(s)}, \Phi_{\beta}^{H(s)}(V) \right\},
    \end{equation}
    where
    \begin{equation}
        \Phi_{\beta}^{H(s)} = \displaystyle\int_{-\infty}^{\infty} dt f_{\beta}(t) e^{- i H(s) t} V e^{i H(s) t} \ \ \ \text{with} \ \ \ f_{\beta}(t) = \frac{1}{2 \pi} \displaystyle\int_{-\infty}^{\infty} d\omega \frac{\tanh(\beta \omega/2)}{\beta \omega/2} e^{i \omega t}.
    \end{equation}
    This immediately implies, via integrating Eq.~\eqref{eq:anti_commute}, that $e^{-\beta (H + V)} = e^{-\beta H} \odot e^{-\beta V} = O e^{-\beta H} O^{\dagger}$, with

\begin{equation}
    \label{eq:o}
    O = \mathcal{T} \exp \left[ -\frac{\beta}{2} \displaystyle\int_{0}^{1} ds' \Phi^{H(s')}_{\beta}(V) \right].
\end{equation}
\end{restatable}

\noindent At a high-level,  Thm.~\ref{thm:error_single_step} utilizes Thm.~\ref{thm:hastings} in order to decompose a thermal state such that partial traces can be propagated to their desired locations. It then makes use of the thermal boundedness property to reason, via a Lieb-Robinson bound, that the original thermal state after perturbation and the effective thermal state obtained after performing the partial trace are sufficiently close. For the full proof of Thm.~\ref{thm:error_single_step}, see Appx.~\ref{appx:proofs}.

\section{Discussion and future work}
\label{sec:conclusion}

\noindent In this work, we aimed to put sliding-window belief propagation on a rigorous footing. We demonstrated that the error accumulated due to propagation of partial traces when computing messages is exponentially-bounded by the size of the chosen sliding-window. There are many future avenues of work which will be necessary in order to tighten this bound for particular instances of quantum belief propagation. Firstly, the notion of thermal boundedness, while well-defined, is somewhat opaque and not well-understood, relative to other properties such as rapidly decaying correlation coefficients. It is still an open problem as to how this property relates, generally or for specific cases, to better-characterized attributes of many-body systems. It would be interesting to develop operators that estimate or witness thermal boundedness or unboundedness, in analogy with entanglement witnesses. Additionally, explicit calculation of the constants associated with the Lieb-Robinson bound utilized in the proof of Thm.~\ref{thm:error_single_step} will vary on a case-by-case basis, so it remains to be seen exactly which kinds of system will exhibit more favourable error scaling. It is important to note that in the limit of $\beta \to \infty$, the derived bound goes to $\infty$ exponentially quickly. We do expect this bound to fail in the zero-temperature limit, but it would be interesting to investigate whether the bound can be improved at low-temperatures, perhaps with more favourable scaling in $\beta$ (if not generically, then for particular classes of systems).

Finally, verification and further study of the results in this paper and some of the future questions being posed via large-scale numerical simulations is likely a promising avenue of further inquiry. Since thermal boundedness is distinct from (though perhaps not completely independent of) long range entanglement, quantum belief propagation and tensor network state representations could potentially be combined to further expand the space of states efficiently simulable on a classical computer. Ultimately, it is our hope that realizing the quantum belief-propagation algorithm and its sliding-window generalization via numerical methods, and possibly in tandem with quantum computers, as was envisioned by Ref.~\cite{leifer2008quantum}, for the purpose of computing properties of large thermal systems, will lead to the discovery and further understanding of the capabilities and limitations of these algorithms on a system-by-system basis, and that rigorous bounds such as those introduced in this work will act as a helping hand in understanding their general behaviour.

\section{Acknowledgements}

\noindent The authors thank Owen Lockwood and Patrick Huembeli for insightful comments and feedback on drafts of this manuscript. J.C. thanks Nathan Wiebe for helpful conversations about Lieb-Robinson bounds and other related topics. Finally, the authors would like to dedicate this paper to the memory of David Poulin, a exceptional scientist, whose foundational contributions made this work possible.

\bibliography{main}

\appendix

\section{Proof of Thm.~\ref{thm:error_single_step}}
\label{appx:proofs}

\noindent We dedicate this section to providing a proof of Thm.~\ref{thm:error_single_step}. Throughout the subsequent derivations, we use $||\cdot||$ (and sometimes $||\cdot||_2$, when we want to be explicit) to denote the operator/spectral norm and $||\cdot||_1$ to denote the trace norm.

\errorsinglestep*

\noindent Recall the theorem characterizing Hastings' belief propagation (we utilize the modified version of Ref.~\cite{kato2019quantum}).

\hastings*

\noindent The proof of Thm.~\ref{thm:error_single_step} essentially reduces to proving two individual bounds, and combining them via a triangle inequality. We separate this section into two subsections, dedicated to each of these bounds, and conclude with their combination as the final, desired bound.

\subsection{The first bound}
\label{app:1}

\noindent Suppose $G = (\mathcal{V}, \mathcal{E})$ is a tree graph, let $v^{*}$ be a leaf-node of the tree. Now, let $H = \sum_{e \in \mathcal{E}} h_e$ be a local Hamiltonian, let $V = H_B \coloneqq H_{B_{\ell}(v^{*})}$ (we call this the Hamiltonian on the \emph{buffer region}). We define $H' = H - V$. We then have, making use of Eq.~\eqref{eq:o},
\begin{equation}
\label{eq:bp}
    \frac{1}{\mathcal{Z}} e^{-\beta H} = \frac{1}{\mathcal{Z}} \left( e^{-\beta H'} \odot e^{-\beta V} \right) = \frac{1}{\mathcal{Z}} O e^{-\beta H'} O^{\dagger} = \frac{1}{\mathcal{Z}} O e^{-\beta H_{L_{\ell}(v^{*})}} e^{-\beta H_{R_{\ell}(v^{*})}} O^{\dagger}.
\end{equation}
where we have
\begin{equation}
e^{-\beta H'} = e^{-\beta [H_{R_{\ell}(v^{*})} + H_{L_{\ell}(v^{*})}]} = e^{-\beta H_{R_{\ell}(v^{*})}} e^{-\beta H_{L_{\ell}(v^{*})}}
\end{equation}
because $H_{R_{\ell}(v^{*})}$ and $H_{L_{\ell}(v^{*})}$ have disjoint support, and therefore commute. Going forward, we simplify notation by letting $H_L = H_{L_{\ell}(v^{*})}$ and $H_R = H_{R_{\ell}(v^{*})}$. It was shown in Ref.~\cite{kato2019quantum} that the operator $O$ of Eq.~\eqref{eq:o} is local, in the sense that it can be locally-approximately in a region with exponentially decaying error in the size of the region. In particular, we have the following result.
\begin{theorem}[Ref.~\cite{kato2019quantum}]
    For operator $O$ of Thm.~\ref{thm:hastings} corresponding to Hamiltonian $H(s) = H + sV$ and inverse temperature $\beta$, there exists $O_{\ell}$ such that 
    \begin{equation}
        ||O - O_{\ell}|| \leq \frac{c \beta ||V||}{2} e^{\frac{(1 + c) \beta ||V||}{2}} e^{-\frac{c\ell}{1 + c\alpha\beta/\pi}}
    \end{equation}
    for positive real numbers $c$ and $\alpha$, where $\mathrm{supp}(O_{\ell})$ is contained in a ball of radius $\ell$ around $\mathrm{supp}(V) \subset \mathcal{V}$. Moreover, $||O_{\ell}||_1 \leq ||O||_1$.
\end{theorem}
\noindent This particular fact does much of the heavy-lifting when it comes to demonstrating the desired error-bound on sliding-window belief-propagation. By this lemma, there exists an $(\ell - 2)$-local approximation of the $O$ of Eq.~\eqref{eq:bp}, $O_{\ell}$, such that
\begin{equation}
        ||O - O_{\ell}|| \leq \frac{c' \beta ||V||}{2} e^{\frac{(1 + c) \beta ||V||}{2}} e^{-\frac{c\ell}{1 + c\alpha\beta/\pi}}
    \end{equation}
where we take $(\ell - 2) \mapsto \ell$ in the exponential via the introduction of constant $c'$. It follows that
\begin{align}
    \left|\left| \frac{1}{\mathcal{Z}} O e^{-\beta H'} O^{\dagger} - \frac{1}{\mathcal{Z}} O_{\ell} e^{-\beta H'} O_{\ell}^{\dagger} \right|\right|_1 &\leq \frac{1}{\mathcal{Z}} \left(  || O e^{-\beta H'} O^{\dagger} -  O e^{-\beta H'} O_{\ell}^{\dagger}||_1 + || O e^{-\beta H'} O_{\ell}^{\dagger} -  O_{\ell} e^{-\beta H'} O_{\ell}^{\dagger} ||_1 \right) \nonumber
    \\ & = \frac{1}{\mathcal{Z}}\left(  ||O e^{-\beta H'} (O^{\dagger} - O_{\ell}^{\dagger}) ||_1 + ||(O - O_{\ell}) e^{-\beta H'} O_{\ell}^{\dagger}||_1\right) \nonumber
    \\ & \leq \left|\left| \frac{e^{-\beta H'}}{\mathcal{Z}} \right|\right|_1 (||O||_2 + ||O_{\ell}||_2)(||O - O_{\ell} ||_2) \nonumber
    \\ & \leq \frac{\mathcal{Z}'}{\mathcal{Z}} c' \beta ||V|| ||O|| e^{\frac{(1 + c) \beta ||V||}{2}} e^{-\frac{c\ell}{1 + c\alpha\beta/\pi}}
    \label{eq:ineq}
\end{align}
where we use a special case of Holder's inequality, in particular $||AB||_1 \leq ||A||_1 ||B||_2$, to get $||ABC||_1 \leq ||A||_1 ||BC||_2 \leq ||A||_1 ||B||_2 ||C||_2$, which gives the desired inequality at the second-last line. $\mathcal{Z} = \text{Tr}[ e^{-\beta H} ]$, $\mathcal{Z}' = \text{Tr}[ e^{-\beta H'} ]$ are partition functions. Note that $O_{\ell}$ is \emph{not} supported on $v^{*}$, as the distance between a node on which $V$ is supported and $v^{*}$ is at least $\ell - 1$, by definition of $O_{\ell}$. From here, it follows that we can propagate a partial trace over $v^{*}$ through $O_{\ell}$ in the following expression:
\begin{equation}
\label{eq:new_thermal}
    \text{Tr}_{v^{*}} \left[ \frac{1}{\mathcal{Z}} O_{\ell} e^{-\beta H'} O^{\dagger}_{\ell} \right] = \frac{1}{\mathcal{Z}} O_{\ell} e^{-\beta H_L} \text{Tr}_{v^{*}}[e^{-\beta H_R}] O_{\ell}^{\dagger} = \frac{1}{\mathcal{Z}} O_{\ell} \left( e^{-\beta H_L} \odot \text{Tr}_{v^{*}}[e^{-\beta H_R}] \right)O_{\ell}^{\dagger}.
\end{equation}
Note that $e^{-\beta H_L} \odot \text{Tr}_{v^{*}}[e^{-\beta H_R}] = \text{Tr}_{v^{*}}[e^{-\beta H'}]$ has the same partition function as $e^{-\beta H'}$, so via an identical chain of inequalities to those of Eq.~\eqref{eq:ineq}, we have
\begin{multline}
\left|\left| \text{Tr}_{v^{*}} \left[ \frac{1}{\mathcal{Z}} O_{\ell} e^{-\beta H'} O^{\dagger}_{\ell} \right] - \frac{1}{\mathcal{Z}} O \left( e^{-\beta H_L} \odot \text{Tr}_{v^{*}}[e^{-\beta H_R}] \right)O^{\dagger} \right|\right|_1 \\
    = \left|\left| \frac{1}{\mathcal{Z}} O_{\ell} \left( e^{-\beta H_L} \odot \text{Tr}_{v^{*}}[e^{-\beta H_R}] \right)O_{\ell}^{\dagger} - \frac{1}{\mathcal{Z}} O \left( e^{-\beta H_L} \odot \text{Tr}_{v^{*}}[e^{-\beta H_R}] \right)O^{\dagger} \right|\right|_1
    \\ \leq \frac{\mathcal{Z}'}{\mathcal{Z}} c' \beta ||V|| ||O|| e^{\frac{(1 + c) \beta ||V||}{2}} e^{-\frac{c\ell}{1 + cv\beta/\pi}}.
    \label{eq:ineq2}
\end{multline}
Note that partial trace decreases trace norm (see Lem.~\ref{lem:trace_norm}), which implies that
\begin{equation}
    \left|\left| \text{Tr}_{v^{*}} \left[ \frac{1}{\mathcal{Z}} O e^{-\beta H'} O^{\dagger}\right] - \text{Tr}_{v^{*}}\left[\frac{1}{\mathcal{Z}} O_{\ell} e^{-\beta H'} O_{\ell}^{\dagger}\right] \right|\right|_1 \leq \left|\left| \frac{1}{\mathcal{Z}} O e^{-\beta H'} O^{\dagger} - \frac{1}{\mathcal{Z}} O_{\ell} e^{-\beta H'} O_{\ell}^{\dagger} \right|\right|_1.
    \label{eq:ineq3}
\end{equation}
Thus, via a triangle inequality applied to Eq.~\eqref{eq:ineq}, Eq.~\eqref{eq:ineq2}, and Eq.~\eqref{eq:ineq3}, we have
\begin{equation}
    \left|\left| \text{Tr}_{v^{*}}\left[ \frac{1}{\mathcal{Z}} e^{-\beta H} \right] - \frac{1}{\mathcal{Z}} O \left( e^{-\beta H_L} \odot \text{Tr}_{v^{*}}[e^{-\beta H_R}] \right)O^{\dagger} \right|\right|_1 \leq \frac{2 c' \beta \mathcal{Z}' ||V|| ||O||}{\mathcal{Z}} e^{\frac{(1 + c) \beta ||V||}{2}} e^{-\frac{c\ell}{1 + c\alpha\beta/\pi}}.
\end{equation}
From here, it remains to bound the partition functions and norms of $O$ and $O_{\ell}$ in terms of known quantities. From Golden-Thompson inequality (Thm.~\ref{thm:golden_thompson}) followed by Holder's inequality,
\begin{equation}
    \frac{\mathcal{Z}'}{\mathcal{Z}} = \frac{\text{Tr}[e^{\beta(V - H)}]}{\text{Tr}[e^{-\beta H}]} \leq \frac{\text{Tr}[e^{-\beta H} e^{\beta V}]}{\text{Tr}[e^{-\beta H}]} = \frac{||e^{-\beta H} e^{\beta V}||_1}{||e^{-\beta H}||_1} \leq \frac{||e^{-\beta H}||_1 || e^{\beta V}||}{||e^{-\beta H}||_1} \leq e^{\beta ||V||}.
\end{equation}
And, of course, by the definition in Thm.~\ref{thm:hastings},
\begin{equation}
\label{eq:o_bnd}
||O|| = \left|\left| \exp \left[ -\frac{\beta}{2} \displaystyle\int_{0}^{1} ds' \Phi^{H(s')}_{\beta}(V) \right] \right|\right| \leq \max_{s' \in [0, 1]} \exp \left( \frac{\beta ||\Phi_{\beta}^{H(s')}(V)||}{2}\right) \leq \exp \left( \frac{\beta ||V||}{2} \displaystyle\int_{-\infty}^{\infty} dt |f_{\beta}(t)| \right),
\end{equation}
where we use $||e^{-iH(s)t} V e^{i H(s) t}|| = ||V||$. Recall that
\begin{align}
f_{\beta}(t) = \frac{1}{\beta \pi} \int_{-\infty}^{\infty} d\omega \frac{1}{\omega} \tanh\left(\frac{\beta \omega}{2} \right) e^{i\omega t}.
\end{align}

\begin{comment}
where $\mathcal{F}^{-1}$ is the inverse Fourier transform. From the convolution theorem,
\begin{align}
    \mathcal{F}^{-1} \left[ \frac{\tanh(\beta \pi t)}{t} \right] &= \mathcal{F}^{-1}(\tanh(\beta \pi t)) \star \mathcal{F}^{-1}(t^{-1}) = \frac{1}{2\pi \beta} \sqrt{\frac{\pi}{2}} \left( \text{csch}\left( \frac{t}{2\beta} \right) \star \text{sign}(t) \right).
\end{align}
Note that
\begin{align}
    \text{csch}\left( \frac{t}{2\beta} \right) \star \text{sign}(t) &= \displaystyle\int_{-\infty}^{\infty} d\tau \text{sign}(\tau) \text{csch} \left( \frac{t - \tau}{2\beta} \right) = \displaystyle\int_{0}^{\infty} d\tau \text{csch} \left( \frac{t - \tau}{2\beta} \right) + \displaystyle\int_{-\infty}^{0} d\tau  \text{csch} \left( \frac{\tau - t}{2\beta} \right)
    \\ & = -2\beta \displaystyle\int_{t/2\beta}^{\infty} d\tau \text{csch}(\tau) + 2\beta \displaystyle\int_{-\infty}^{-t/2\beta} d\tau \text{csch}(\tau) = -4\beta \displaystyle\int_{t/2\beta}^{\infty} d\tau \text{csch}(\tau)
    \\ & = 4\beta \log \left| \tanh\left(\frac{t}{4\beta} \right)\right|.
\end{align}
This implies that
\begin{equation}
    \label{eq:kernel}
    f_{\beta}(t) = \frac{1}{\beta \pi} \sqrt{\frac{2}{\pi}} \log \left| \tanh \left( \frac{t}{4\beta} \right) \right|.
\end{equation}
\end{comment}

\begin{lemma}
For all $t \in \mathbb{R} - \{0\}$, we have
\begin{equation}
    |f_{\beta}(t)| = \frac{1}{\beta \pi} \left| \int_{-\infty}^{\infty} d\omega \frac{1}{\omega} \tanh\left(\frac{\beta \omega}{2} \right) e^{i\omega t}\right| = \frac{2}{\beta \pi} \left| \log \left( \coth \left( \frac{\pi |t|}{2\beta} \right) \right) \right|.
\end{equation}
\end{lemma}
\begin{proof}
Evaluating this integral can be reduced to a routine application of the residue calculus. Note that
\begin{equation}
    g(\omega) = \frac{1}{\omega} \tanh \left( \frac{\beta \omega}{2} \right) = \frac{\left( e^{\beta \omega / 2} - e^{- \beta \omega / 2} \right)}{\omega \left( e^{\beta \omega / 2} + e^{-\beta \omega / 2} \right)}
\end{equation}
is meromorphic, as it is the quotient of holomorphic functions. This function clearly has a removable singularity at $\omega = 0$, so we can extend it to a holomorphic function at $0$. This function will have poles precisely when
\begin{align}
   e^{\beta \omega / 2} + e^{-\beta \omega / 2} = 0 &\Longleftrightarrow e^{\beta \omega} = -1 \\ & \Longleftrightarrow \omega = p_k = \frac{i (2k + 1) \pi}{\beta} \ \ \ \text{for} \ k \in \mathbb{Z}.
\end{align}
Each of these poles are simple as
\begin{equation}
    \frac{d}{d\omega} \left( e^{\beta \omega / 2} + e^{-\beta \omega / 2} \right) \biggr\rvert_{\omega = p_k} = \frac{\beta}{2} \left( e^{(2 k + 1) \pi i / 2} - e^{-(2k + 1) \pi i / 2} \right) \neq 0.
\end{equation}
for each $k$. Thus, it follows immediately that
\begin{align}
    \text{Res} \left( g(\omega) e^{i \omega t}, p_k \right) &= \frac{\left( e^{\beta p_k / 2} - e^{- \beta p_k / 2} \right) e^{i p_k t}}{ \frac{p_k \beta}{2} \left( e^{\beta p_k / 2} - e^{-\beta p_k / 2} \right) + \left( e^{\beta p_k / 2} + e^{-\beta p_k / 2} \right)}
    \\ & = \frac{2 e^{i p_k t}}{\beta p_k} = -\frac{2 i}{(2k + 1) \pi} e^{- (2k + 1) t \pi / \beta}
\end{align}
From this result, note in particular that the partial sums
\begin{equation}
    S_n = 2\pi i \displaystyle\sum_{k = 0}^{n} \text{Res}\left( g(\omega) e^{i \omega t}, p_k \right) = \displaystyle\sum_{k = 0}^{n} \frac{4}{2k + 1} e^{-(2k + 1) t \pi / \beta}
\end{equation}
converge for $n \to \infty$ when $t > 0$: the sum is bounded by the sums over a decaying exponential.

With knowledge of the residues, we can carry out the integral. We choose as our integration contour the semi-circle lying in the upper-half plane and the $x$-axis, symmetric about the $y$-axis, with radius $R$, when $t > 0$. As $R \to \infty$, one can see that the portion of the integral on the circle's arc will go to $0$. The singularities of $g(\omega)$ contained inside the half-disk region enclosed by the semi-circle are $p_k$ for $k \geq 0$. Thus, since the partial sums $S_n$ converge, we find from residue theorem that for $t > 0$,
\begin{equation}
    \displaystyle\int_{-\infty}^{\infty} d\omega \frac{1}{\omega} \tanh\left( \frac{\beta \omega}{2} \right) e^{i \omega t} = \displaystyle\sum_{k = 0}^{\infty} \frac{4}{2k + 1} e^{-(2k + 1) t \pi / \beta} = 2 \log \left( \frac{1 + e^{-\pi t/\beta}}{1 - e^{-\pi t/\beta}} \right) = 2 \log \left( \coth \left( \frac{\pi t}{2\beta} \right) \right)
\end{equation}
Finally, note that $f_{\beta}(t)$ is even (this follows immediately from the definition). The desired result is an immediate consequence.
\end{proof}
\noindent From this result, it is then possible to conclude that
\begin{align}
    \displaystyle\int_{-\infty}^{\infty} |f_{\beta}(t)| = \frac{2}{\beta \pi} \displaystyle\int_{-\infty}^{\infty} dt \left| \log \left( \coth \left( \frac{\pi |t|}{2\beta} \right) \right) \right| = \frac{8}{\pi^2} \displaystyle\int_{0}^{\infty} \log(\coth(t)) \ dt = 1.
\end{align}
This immediately implies, using Eq.~\eqref{eq:o_bnd}, that $||O|| \leq e^{\beta ||V||/2}$. Putting everything together,
\begin{equation}
\label{eq:bnd_first}
    \left|\left| \text{Tr}_{v^{*}}\left[ \frac{1}{\mathcal{Z}} e^{-\beta H} \right] - \frac{1}{\mathcal{Z}} O \left( e^{-\beta H_L} \odot \text{Tr}_{v^{*}}[e^{-\beta H_R}] \right)O^{\dagger} \right|\right|_1 \leq 2c' \beta ||V|| e^{\frac{(4 + c) \beta ||V||}{2}} e^{-\frac{c\ell}{1 + c\alpha\beta/\pi}}.
\end{equation}
We have accomplished the first step of our bounding procedure: we have managed to move the partial trace in front of the thermal state to a location ``within" the agglomeration of local circle products which make up the entire global state, at the cost of introducing same error \emph{which decays exponentially in $\ell$}.

\subsection{The second bound}
\label{app:2}

\noindent The second and final task is to demonstrate that $O$ is a sufficient approximation of the true operator $O_{\text{eff}}$ which induces the removed Hamiltonian terms on the buffer region, $V$, for the base state $e^{-\beta H_L} \odot \text{Tr}_{v^{*}}[e^{-\beta H_R}]$. To be more concrete, we know, under the assumption that $e^{-\beta H_L} \odot \text{Tr}_{v^{*}}[e^{-\beta H_R}]$ is non-singular, that $e^{-\beta H_L} \odot \text{Tr}_{v^{*}}[e^{-\beta H_R}] = e^{-\beta H_{\text{eff}}}$ for an effective Hamiltonian  $H_{\text{eff}}$. Via Hastings' belief-propagation, there exists operator $O_{\text{eff}}$ such that $O_{\text{eff}} e^{-\beta H_{\text{eff}}} O_{\text{eff}}^{\dagger} = e^{-\beta (H_{\text{eff}} + V)}$. We wish to determine the closeness of $O$ and $O_{\text{eff}}$. The effective Hamiltonian $H_{\text{eff}}$ can be written as
\begin{align}
    H_{\text{eff}} &= H_{L} - \frac{1}{\beta} \log \text{Tr}_{v^{*}}\left[e^{-\beta H_R} \right] = H' + \left( -\frac{1}{\beta} \log \text{Tr}_{v^{*}}\left[e^{-\beta H_R}\right] - H_R \right)
\end{align}
Clearly, the second term is simply the thermal potential of $H_R$ effectuated by $v^{*}$, as was introduced in Def.~\ref{def:thermal_potential}, $V_{\text{thermal}}(H_R; v^{*}, \beta)$. It is at this point that we will require the thermal boundedness property, which we restate for convenience.

\thermalbnd*

\noindent We are now able to prove the desired bound on the distance between $O$ and $O_{\text{eff}}$, if we assume that the underlying Hamiltonian has this property. This will, however, require some auxiliary results, which we now present.

\begin{remark}[Long-time Lieb-Robinson bound]
\label{rem:lieb_robinson_piece}
When a local Hamiltonian $H$ supported on the vertices of a graph $G = (\mathcal{V}, \mathcal{E})$ with Hilbert space $\mathcal{H} = \bigotimes_{w \in \mathcal{V}} \mathcal{H}^{w}$ is said to satisfy a \emph{Lieb-Robinson bound}, it means that there exist constants $C, a, v \in \mathbb{R}^{+}$ such that for normalized observables $A$ and $B$ with supports in $\mathcal{H}$ separated by distance $\ell$,
\begin{equation}
    ||[A(t), B]|| \coloneqq || [e^{i H t} A e^{- i H t}, B] || \leq C e^{-a\ell} e^{av|t|}
\end{equation}
for all $t \in \mathbb{R}$. Fixing $\ell$, as $|t| \to \infty$, the Lieb-Robinson bound also approaches infinity. However, note that
\begin{equation}
    || [e^{i H t} A e^{- i H t}, B] || = || || e^{i H t} A e^{- i H t} B - B e^{i H t} A e^{- i H t} || \leq 2||A||||B|| = 2,
\end{equation}
a constant bound. Clearly, for $|t|$ sufficiently large, this naive, constant bound will be an improvement over the Lieb-Robinson bound. Thus, when using a Lieb-Robinson bound for long-time dynamics, it is important to tighten it, via the constant bound, as
\begin{equation}
    ||[A(t), B]|| \leq \min \left\{ C e^{-a\ell} e^{av|t|}, 2 \right\}.
\end{equation}
\end{remark}

\noindent This next result effectively will demonstrate that local Hamiltonian perturbations induce local effects on time-evolution of local observables. Suppose we are given an observable $O$ which is supported on a particular local region, as well as some Hamiltonian under which we time-evolve the system. Suppose $V$ is a Hamiltonian perturbation which has ``exponentially decaying influence" away from another region, which is separated from the support of $O$ by some distance $\ell$. We demonstrate that for sufficiently short times, evolving $O$ under $H + V$ yields very similar dynamics to evolution of $O$ under the bare Hamiltonian $H$. In particular, we have the following:

\begin{lemma}[Localization of time-evolution]
\label{lem:local_time}
    Let $H$ be a local Hamiltonian, $H = \sum_{e \in \mathcal{E}} h_{e}$ satisfying a Lieb-Robinson bound on the graph $G$, let $V$ be a perturbation having the property that $||V^{(j)}|| \leq K e^{-k j}$, for $K, k \in \mathbb{R}^{+}$, where cumulants $V^{(j)}$ are taken with respect to the region $\mathcal{V}' \subset \mathcal{V}$.
 
    Let $H' = H + V$. Let $O$ be a normalized observable supported on region $\mathcal{V}''$ such that $d(\mathcal{V}', \mathcal{V}'') = \ell$. Then,
    \begin{equation}
        || e^{i H' t} O e^{-i H' t} - e^{i H t} O e^{-i H t} || \leq \min \left\{ \left( |t| L + \frac{C K e^{avt}}{av} \right) e^{-a' \ell}, |t| L \right\}
    \end{equation}
    for fixed constants $a$, $v$, $C$, $a' = \min(k, a)$, and $L = 2K(1 - e^{-k})^{-1}$.
\end{lemma}
\begin{proof}
    This lemma relies on use of the Lieb-Robinson bound that is assumed to hold, from which the constants $a$, $v$, and $C$ originate. In particular, for normalized observables $A$ and $B$ supported on subregions of $\mathcal{V}$ separated by distance $l$, we have from Rem.~\ref{rem:lieb_robinson_piece} that
    \begin{equation}
        || [e^{i H t} A e^{- i H t}, B] || \leq \min \left\{
            C e^{-al} e^{av|t|}, 2 \right\}
    \end{equation}
    Note that $V = \sum_{j = 0}^{\infty} V^{(j)}$, where $V^{(j)}$ is supported entirely on terms lying a distance of at most $j$ from $\mathcal{V}'$. Therefore, the distance between the support of $V^{(j)}$ and $\mathcal{V}''$ is at least $\ell - j$. It follows that
    \begin{align}
        \label{eq:bnd}
         || [e^{i H t} O e^{- i H t}, V^{(j)}] || = ||V^{(j)}|| \left|\left| \left[e^{i H t} O e^{- i H t}, \frac{V^{(j)}}{||V^{(j)}||} \right] \right|\right|  &\leq \min \left\{ C K e^{-kj} e^{-a(\ell - j)} e^{av|t|}, 2 K e^{-kj} \right\} \nonumber
         \\ & \leq  \min \left\{ C K e^{-a' \ell } e^{av|t|}, 2 K e^{-kj} \right\}
    \end{align}
    where $a' = \min(k, a)$. Note that of the two individual upper-bounds of which we take the minimum in the above inequality, one decays exponentially in $\ell$, but increases exponentially in $|t|$, while the other is constant in both parameters. It follows, similar to above, if we fix $\ell$ and allow $|t| \to \infty$, then the first bound will become arbitrarily large, whereas this will not be the case with the constant bound. Thus, to bound bound the commutator of the time-evolved $O$ with $V$, we construct an upper-bound depending on $t$ and $\ell$, as well as a constant upper-bound. In the first case, note that
    \begin{align}
        || [e^{i H t} O e^{- i H t}, V] || &\leq \displaystyle\sum_{j = 0}^{\ell - 1} || [e^{i H t} O e^{- i H t}, V^{(j)}] || + \displaystyle\sum_{j = \ell}^{\infty} || [e^{i H t} O e^{- i H t}, V^{(j)}] ||
        \\ & \leq K C e^{-a' \ell } e^{av|t|} + 2 K \displaystyle\sum_{j = \ell}^{\infty} e^{-kj}
        \\ & \leq K C e^{-a' \ell } e^{av|t|} + 2 K e^{-k\ell} \displaystyle\sum_{j = 0}^{\infty} e^{-k j}
        \\ & \leq K \left( C e^{av|t|} + \frac{2}{1 - e^{-k}} \right) e^{-a' \ell}
    \end{align}
    and in the latter case,
    \begin{equation}
        || [e^{i H t} O e^{- i H t}, V] || \leq 2K \displaystyle\sum_{j = 0}^{\infty} e^{-kj} = \frac{2K}{1 - e^{-k}}.
    \end{equation}
    Therefore,
    \begin{align}
        || [e^{i H t} O e^{- i H t}, V] || &\leq \min \left\{ K \left( C e^{av|t|} + \frac{2}{1 - e^{-k}} \right) e^{-a' \ell}, \frac{2K}{1 - e^{-k}}\right\}
        \\ & = \min \left\{ (C K e^{av|t|} + L) e^{-a'\ell}, L \right\}
    \end{align}
    where we have defined $L = 2K (1 - e^{-k})^{-1}$. From here,
        \begin{align}
        \label{eq:big_eq}
        || e^{i V s} e^{i H t} O e^{- i H t} e^{-i V s} -  e^{i H t} O e^{- i H t} || &= ||[ e^{i H t} O e^{- i H t}, e^{iVs}]|| = \left|\left| \displaystyle\sum_{n = 0}^{\infty} \frac{1}{n!} \left[ e^{i H t} O e^{- i H t}, \left( iVs \right)^n \right] \right|\right| \nonumber
        \\ & \leq \displaystyle\sum_{n = 0}^{\infty} \frac{|s|^n}{n!} || [e^{i H t} O e^{- i H t}, V^n] || \leq \displaystyle\sum_{n = 1}^{\infty} \frac{||V||^{n - 1} |s|^n}{(n - 1)!} ||[e^{i H t} O e^{- i H t}, V]|| \nonumber
        \\ & = |s| ||[e^{i H t} O e^{- i H t}, V]|| \displaystyle\sum_{n = 0}^{\infty} \frac{||V||^n |s|^n}{n!} = |s| ||[e^{i H t} O e^{- i H t}, V]|| e^{|s| ||V||} \nonumber
        \\ &\leq |s| \min \left\{ (C K e^{av|t|} + L) e^{-a' \ell}, L \right\} e^{|s|||V||}.
    \end{align}
    where we make use of Lem.~\ref{lem:commutator_bound} to bound the commutator with powers of $V$. Now, recall the first-order Trotter formula, which states that
    \begin{equation}
        e^{i H' t} = e^{i (H + V) t} = \lim_{n \to \infty} \left( e^{i V t / n} e^{i H t / n} \right)^{n} \coloneqq \lim_{n \to \infty} U_n(t; H, V).
    \end{equation}
    We inductively bound the error accumulated via Trotter decomposition. In particular, applying Lem.~\ref{lem:telescoping}, as well as Eq.~\ref{eq:big_eq}, we have
    \begin{align}
        \left|\left| U_n(t; H, V) O U_n(t; H, V)^{\dagger} - e^{i H t} O e^{-i H t} \right|\right| &\leq \displaystyle\sum_{j = 1}^{n} \left|\left| e^{i H j t / n} O e^{-i H j t / n} - e^{i V t / n}  e^{i H j t / n} O e^{-i H j t / n} e^{-i V t / n} \right|\right| \nonumber
        \\ & \leq \displaystyle\sum_{j = 1}^{n} \frac{|t|}{n} \min \left\{ (C K e^{avj|t|/n} + L) e^{-a'\ell}, L \right\} e^{|t| ||V||/n} \nonumber
        \\ & = |t| e^{|t| ||V||/n} \min \left\{ \left( L + \frac{C K}{n} \displaystyle\sum_{j = 1}^{n} e^{avj|t|/n} \right) e^{-a' \ell}, L \right\} \nonumber
        \\ & \leq |t| e^{|t| ||V||/n} \min \left\{ \left( L + \frac{C K}{n} \frac{e^{av|t|}}{e^{av|t|/n} - 1} \right) e^{-a' \ell}, L \right\}
    \end{align}
    Letting $f(x) = \frac{e^{q/x}}{x(e^{r/x} - 1)}$, note that
   \begin{align}
       \lim_{x \to \infty} f(x) = \lim_{x \to \infty} \frac{e^{q/x}}{x(e^{r/x} - 1)} = \lim_{x \to 0} \frac{x e^{qx}}{e^{rx} - 1} = \lim_{x \to 0} \frac{e^{qx} + q x e^{qx}}{r e^{rx}} = \frac{1}{r}
   \end{align}
   which finally implies that
   \begin{align}
       || e^{i H' t} O e^{-i H' t} - e^{i H t} O e^{-i H t} || &= \lim_{n \to \infty} \left|\left| U_n(t; H, V) O U_n(t; H, V)^{\dagger} - e^{i H t} O e^{-i H t} \right|\right|
       \\ & \leq \min \left\{ \left( |t| L + \frac{C K e^{av|t|}}{av} \right) e^{-a' \ell}, |t| L \right\}
   \end{align}
   This completes the proof.
\end{proof}
\noindent From here, we return to bounding the difference between $O$ and $O_{\text{eff}}$. Indeed, note from Lem.~\ref{lem:exp_bnd} that
\begin{align}
\label{eq:bnd2}
    ||O - O_{\text{eff}}|| &= \left|\left| \exp \left[ -\frac{\beta}{2} \displaystyle\int_{0}^{1} ds' \Phi^{H'(s')}_{\beta}(V) \right] - \exp \left[ -\frac{\beta}{2} \displaystyle\int_{0}^{1} ds' \Phi^{H_{\text{eff}}(s')}_{\beta}(V) \right] \right|\right| \nonumber
    \\ & \leq \frac{\beta e^{M}}{2} \left|\left| \displaystyle\int_{0}^{1} ds' \Phi^{H'(s')}_{\beta}(V) - \displaystyle\int_{0}^{1} ds' \Phi^{H_{\text{eff}}(s')}_{\beta}(V) \right|\right| \leq \frac{\beta e^{M}}{2} \displaystyle\int_{0}^{1} ds' \left|\left| \Phi^{H'(s')}_{\beta}(V) - \Phi^{H_{\text{eff}}(s')}_{\beta}(V) \right|\right| \nonumber
    \\ & \leq \frac{\beta e^{M}}{2} \displaystyle\int_{0}^{1} ds' \displaystyle\int_{-\infty}^{\infty} dt |f_{\beta}(t)| \left|\left| e^{- i H'(s) t} V e^{i H'(s) t} - e^{- i H_{\text{eff}}(s) t} V e^{i H_{\text{eff}}(s) t} \right|\right|,
\end{align}
with
\begin{equation}
    M = \frac{\beta}{2} \max \left\{ \left| \left| \displaystyle\int_{0}^{1} ds' \Phi^{H'(s')}_{\beta}(V) \right|\right|, \left| \left| \displaystyle\int_{0}^{1} ds' \Phi^{H_{\text{eff}}(s')}_{\beta}(V) \right|\right| \right\} \leq \frac{\beta ||V||}{2},
\end{equation}
where this upper-bound is computed using the bound derived earlier on the integral over $|f_{\beta}(t)|$. Now, recall that $H_{\text{eff}} = H' + V_{\text{thermal}}(H_R; v^{*}, \beta)$. We assumed that the Hamiltonian $H$ is thermally bounded, $H_R$ is a sum of a subset of the terms of $H$ and $v^{*}$ is a subset of $\mathcal{V}$. Thus, $||V^{(j)}_{\text{thermal}}(H_R; v^{*}, \beta)|| \leq K(\beta) e^{-k(\beta) j}$. We have
\begin{equation}
    H'(s) = H' + sV \ \ \ \text{and} \ \ \ H_{\text{eff}}(s) = H_{\text{eff}} + sV = H'(s) + V_{\text{thermal}}(H_R; v^{*}, \beta).
\end{equation}
Thus, we can invoke Lem.~\ref{lem:local_time} to get
\begin{align}
   \left|\left| e^{- i H'(s) t} V e^{i H'(s) t} - e^{- i H_{\text{eff}}(s) t} V e^{i H_{\text{eff}}(s) t} \right|\right| &\leq \min \left\{ \left( |t| L(\beta) + \frac{C K(\beta) e^{av|t|}}{av} \right) e^{-a'(\beta) \ell}, |t| L(\beta) \right\}
   \\ & \coloneqq \min \left\{ \left( |t| L(\beta) + M(\beta) e^{av|t|} \right) e^{-a'(\beta) \ell}, |t| L(\beta) \right\}
\end{align}
where we introduce $\beta$-dependence on all constants which themselves depend on $\beta$ via $k(\beta)$ and $K(\beta)$, and maximize over $s \in [0, 1]$ to get universal constants in the bound, rather than ones which depend on $s$. Note that
\begin{equation}
    k(\beta) \in O(\beta^{-1}) \Longrightarrow \frac{1}{1 - e^{-k(\beta)}} \in O(k(\beta)^{-1}) = O(\beta)
\end{equation}
which means that $L(\beta), M(\beta) \in O(\exp(\beta))$. Immediately, from Eq.~\ref{eq:bnd2},
\begin{equation}
    \label{eq:tough_bound}
    ||O - O_{\text{eff}}|| \leq \frac{\beta e^{\beta ||V||/2}}{2} \displaystyle\int_{-\infty}^{\infty} dt |f_{\beta}(t)| \min \left\{ \left( |t| L(\beta) + M(\beta) e^{av|t|} \right) e^{-a'(\beta) \ell}, |t| L(\beta) \right\},
\end{equation}
where we eliminate the integration over $s'$ by maximizing over $s' \in [0, 1]$. We now can use the previously derived formula for $|f_{\beta}(t)|$ to get
\begin{multline}
    \label{eq:bnd3}
    \displaystyle\int_{-\infty}^{\infty} dt |f_{\beta}(t)| \min \left\{ \left( |t| L(\beta) + M(\beta) e^{av|t|} \right) e^{-a'(\beta) \ell}, |t| L(\beta) \right\} \\ = \frac{2}{\beta \pi} \displaystyle\int_{-\infty}^{\infty} dt \left| \log \left( \coth \left( \frac{\pi |t|}{2\beta} \right) \right) \right| \min \left\{ \left( |t| L(\beta) + M(\beta) e^{av|t|} \right) e^{-a'(\beta) \ell}, |t| L(\beta) \right\}
    \\ = \frac{4}{\beta \pi} \displaystyle\int_{0}^{\infty} dt \log \left( \coth \left( \frac{\pi t}{2\beta} \right) \right) \min \left\{ \left( L(\beta) t + M(\beta) e^{av|t|} \right) e^{-a'(\beta) \ell}, L(\beta) t \right\}.
\end{multline}
First, it is straightforward to verify that
\begin{equation}
    \displaystyle\int_{0}^{\infty} dt \log \left( \coth \left( \frac{\pi t}{2\beta} \right) \right) t = \frac{4 \beta^2}{\pi^2} \displaystyle\int_{0}^{\infty} dt \log(\coth(t)) t = \frac{7 \zeta(3) \beta^2}{4 \pi^2} < \frac{9\beta^2}{4 \pi^2}.
\end{equation}
Now, let us consider the term in the integral which grows exponentially in $t$. Generally speaking, $\log(\coth(x))$ will go to $0$ as $x \to \infty$ exponentially quickly. However, $e^{x}$ will blow up to infinity at the same rate, so in the case that $av$ is too large, the integral of $f_{\beta}(t)$ against $e^{avt}$ over $t \in \mathbb{R}^{+}$ will not converge. However, when $t$ becomes large enough, it is clear that the function $L(\beta) t$ growing \emph{linearly} in time rather than exponentially will be the least element chosen by the minimum, inside the integral. Thus, the integral over the minimum of Eq.~\eqref{eq:bnd3} will always converge. Despite this, to upper-bound this integral by a quantity which decays exponentially in $\ell$, as we desire, it is necessary to partition this integral in a particular fashion, such that only the tail is integrated against the linear function.

To be more specific, for $x \in [1, \infty)$,
\begin{align}
    \log(\coth(x)) = \log \left( \frac{e^x + e^{-x}}{e^x - e^{-x}} \right) &= \log(e^{x} + e^{-x}) - \log(e^{x} - e^{-x}) = \log(1 + e^{-2x}) - \log(1 - e^{-2x})
    \\ & \leq 2 \log \left(1 - e^{-2x} \right) \leq 4 e^{-2x}
\end{align}
Therefore, it follows that for some $D$, we have
\begin{align}
    I_1(D) \coloneqq \displaystyle\int_{0}^{D} dt \log \left( \coth \left( \frac{\pi t}{2\beta} \right) \right) e^{avt} & = \frac{2 \beta}{\pi} \displaystyle\int_{0}^{\pi D/2\beta} dt \ \log(\coth(t)) e^{2\beta a v t / \pi} \nonumber
    \\ & = \frac{2\beta}{\pi} \displaystyle\int_{0}^{1} dt \ \log(\coth(t)) e^{2\beta a v t / \pi} + \frac{2\beta}{\pi} \displaystyle\int_{1}^{\pi D / 2\beta} dt \ \log(\coth(t)) e^{2\beta a v t / \pi} \nonumber
    \\ & < \frac{8 \beta}{\pi} e^{2\beta av / \pi} + \frac{8\beta}{\pi} \displaystyle\int_{1}^{\pi D/2\beta} dt \ e^{-2t} e^{2\beta a v t / \pi} \nonumber
    \\ & \leq \frac{8 \beta}{\pi} e^{2\beta av / \pi} + \frac{8\beta}{\pi} \left( \displaystyle\int_{0}^{\pi D / 2\beta} e^{4 \beta av t / \pi} \right)^{1/2} \left( \displaystyle\int_{0}^{\pi D / 2\beta} e^{-4t} \right)^{1/2}
    \\ & \leq \frac{8\beta}{\pi} \left[ e^{2\beta av/\pi} + \frac{1}{4} \sqrt{\frac{\pi}{\beta av}} e^{avD} \right]
\end{align}
which is exponentially increasing in $D$. We use the integral Cauchy-Schwarz inequality on the second-to-last line above. Note that this is not the best possible bound we can achieve: if $av\beta$ is small, then we can achieve a constant bound, but we're only interested in the worst case, for now. We also define,
\begin{equation}
    I_2(D) \coloneqq \displaystyle\int_{D}^{\infty} dt \log \left( \coth \left( \frac{\pi t}{2\beta} \right) \right) t \leq 4 \displaystyle\int_{D}^{\infty} dt e^{-\pi t/\beta} t = 4\left(\frac{\beta D}{\pi} + \frac{\beta^2}{\pi^2} \right) e^{-\pi D/\beta}.
\end{equation}
which is (asymptotically) exponentially decreasing in $D$. Note that the integral of Eq.~\ref{eq:bnd3} can be written as $I = I_1(D) + I_2(D)$, for \emph{any} $D \in [1, \infty)$. It is now our task to choose $D$ such that we may bound the integral $I$ with a quantity which exhibits exponential decay in $\ell$. 

In order to do this, we set $D = r \ell$. First, for any choice of $r \in \mathbb{R}^{+}$, we will have
\begin{align}
\label{eq:bnd4}
    I &= \displaystyle\int_{0}^{\infty} dt \log \left( \coth \left( \frac{\pi t}{2\beta} \right) \right) \min \left\{ \left( L(\beta) t + M(\beta) e^{avt} \right) e^{-a'(\beta) \ell},  L(\beta)t \right\} \nonumber
    \\ & = L(\beta) e^{-a'(\beta) \ell} \displaystyle\int_{0}^{\infty} dt \ \log \left( \coth \left( \frac{\pi t}{2\beta} \right) \right) t + \displaystyle\int_{0}^{\infty} dt \log \left( \coth \left( \frac{\pi t}{2\beta} \right) \right) \min \left\{ M(\beta) e^{-a'(\beta) \ell} e^{avt}, (1 - e^{-a'(\beta) \ell} ) L(\beta) t \right\} \nonumber \\
    & < \frac{9 L(\beta) \beta^2 e^{-a'(\beta) \ell}}{4\pi^2} + \displaystyle\int_{0}^{\infty} dt \log \left( \coth \left( \frac{\pi t}{2\beta} \right) \right) \min \left\{ M(\beta) e^{-a'(\beta) \ell} e^{avt}, (1 - e^{-a'(\beta) \ell} ) L(\beta) t \right\} \nonumber
    \\ & \leq 9 L(\beta) \beta^2 e^{-a'(\beta) \ell} + M(\beta) e^{-a'(\beta)\ell} I_1(r\ell) + L(\beta) I_2(r\ell).
\end{align}
Of course, it is possible to choose an optimal value of $r$, in terms of the other constants involved in the final expression of Eq.~\ref{eq:bnd4}. However, to keep things simple, we will simply choose $r$ such that $av r < a'(\beta)$, which will ensure the desired exponential decay. In particular, we set $r = a'(\beta) / 2av$. It follows that
\begin{equation}
    9 L(\beta) \beta^2 e^{-a'(\beta) \ell} + M(\beta) e^{-a'(\beta) \ell} I_1(r\ell)\leq \left[ 9L(\beta) \beta^2 + \frac{8\beta}{\pi}  e^{2\beta av/\pi} + 2\sqrt{\frac{\beta}{\pi av}} \right] e^{-a'(\beta) \ell / 2} \coloneqq \widetilde{M}(\beta) e^{-a'(\beta)\ell/2}
\end{equation}
and
\begin{equation}
    L(\beta) I_2(r\ell) \leq 4 L(\beta) \left( \frac{\beta a'(\beta) \ell}{2av} + \frac{\beta^2}{\pi^2} \right) e^{-\pi a'(\beta) \ell / 2av \beta} \coloneqq (\ell \widetilde{L}_1(\beta) + \widetilde{L}_2(\beta)) e^{-\pi a'(\beta) \ell / 2av \beta}.
\end{equation}
Thus, putting everything together,
\begin{equation}
    I \leq \widetilde{M}(\beta) e^{-a'(\beta) \ell / 2} + (\ell \widetilde{L}_1(\beta) + \widetilde{L}_2(\beta)) e^{-\pi a'(\beta) \ell / 2av \beta} \leq (\ell F(\beta) + G(\beta)) e^{-f(\beta) \ell}
\end{equation}
where $f(\beta) = \min \{ 1/2, \pi/2\beta av\}$, and $F$ and $G$ are functions of $\beta$ (as well as $C$, $a$, $v$, implicitly). It follows immediately from Eq.~\eqref{eq:tough_bound} that
\begin{equation}
\label{eq:bnd_second}
    ||O - O_{\text{eff}}|| \leq \frac{\beta e^{\beta ||V||/2}}{2} (\ell F(\beta) + G(\beta)) e^{-f(\beta) \ell}.
\end{equation}
Therefore, as desired, we have demonstrated that the operator $O$ and the operator $O_{\text{eff}}$ can be upper-bounded by a quantity which decays exponentially in $\ell$, the size of the chosen ``sliding-window".

\subsection{Completing the proof: combining the bounds}

\noindent With the bounds of Sec.~\ref{app:1} and Sec.~\ref{app:2}, we can complete the proof of Thm.~\ref{thm:error_single_step}. Because we now know the difference between $O$ and $O_{\text{eff}}$, we are able to bound the following error:
\begin{multline}
\label{eq:final_f}
\left|\left|\frac{1}{\mathcal{Z}} O \left( e^{-\beta H_L} \odot \text{Tr}_{v^{*}}[e^{-\beta H_R}] \right)O^{\dagger} - \frac{1}{\mathcal{Z}} \left( e^{-\beta (H_L + H_B)} \odot \text{Tr}_{v^{*}}[e^{-\beta H_R}] \right) \right|\right|_1 \\
    = \left|\left|\frac{1}{\mathcal{Z}} O \left( e^{-\beta H_L} \odot \text{Tr}_{v^{*}}[e^{-\beta H_R}] \right)O^{\dagger} - \frac{1}{\mathcal{Z}} O_{\text{eff}} \left( e^{-\beta H_L} \odot \text{Tr}_{v^{*}}[e^{-\beta H_R}] \right)O_{\text{eff}}^{\dagger} \right|\right|_1
    \\ \leq \left|\left| \frac{e^{-\beta H_L} \odot \text{Tr}_{v^{*}}[e^{-\beta H_R}]}{\mathcal{Z}} \right|\right|_1 (||O||_2 + ||O_{\text{eff}}||_2)(||O - O_{\text{eff}} ||_2) 
    \\ \leq \frac{\beta}{2} e^{\beta ||V||} (\ell F(\beta) + G(\beta)) \left|\left| \frac{e^{-\beta H_L} \odot \text{Tr}_{v^{*}}[e^{-\beta H_R}]}{\mathcal{Z}} \right|\right|_1 e^{-f(\beta) \ell}.
\end{multline}
where we once again use Holder's inequality, as we did in Sec.~\ref{app:1}, as well as the main bound derived in Sec.~\ref{app:2}, summarized in Eq.~\ref{eq:bnd_second}. Now, note that
\begin{align}
    \left|\left| \frac{e^{-\beta H_L} \odot \text{Tr}_{v^{*}}[e^{-\beta H_R}]}{\mathcal{Z}} \right|\right|_1 = \frac{\text{Tr} \left[ e^{-\beta (H_L + H_R + V - V)} \right]}{\text{Tr} \left[ e^{-\beta (H_L + H_R + V)} \right]} \leq \frac{\left|\left| e^{-\beta (H_L + H_R + V)} e^{\beta V} \right| \right|_1}{\left|\left| e^{-\beta (H_L + H_R + V)} \right|\right|_1} \leq e^{\beta ||V||}
\end{align}
where we again use Golden-Thompson and Holder's inequalities. To conclude, we use the main bound derived in Sec.~\ref{app:1}, summarized in Eq.~\ref{eq:bnd_first}, and a triangle inequality with Eq.~\eqref{eq:final_f}, to get
\begin{multline}
    \left|\left| \text{Tr}_{v^{*}}\left[ \frac{1}{\mathcal{Z}} e^{-\beta H} \right] - \frac{1}{\mathcal{Z}} \left( e^{-\beta (H_L + H_B)} \odot \text{Tr}_{v^{*}}[e^{-\beta H_R}] \right) \right|\right|_1
    \leq \left|\left| \text{Tr}_{v^{*}}\left[ \frac{1}{\mathcal{Z}} e^{-\beta H} \right] - \frac{1}{\mathcal{Z}} O \left( e^{-\beta H_L} \odot \text{Tr}_{v^{*}}[e^{-\beta H_R}] \right)O^{\dagger} \right|\right|_1 \\ + \left|\left|\frac{1}{\mathcal{Z}} O \left( e^{-\beta H_L} \odot \text{Tr}_{v^{*}}[e^{-\beta H_R}] \right)O^{\dagger} - \frac{1}{\mathcal{Z}} \left( e^{-\beta (H_L + H_B)} \odot \text{Tr}_{v^{*}}[e^{-\beta H_R}] \right) \right|\right|_1
    \\ \leq 2 \beta c||V|| e^{\frac{(4 + c) \beta ||V||}{2}} e^{-\frac{c\ell}{1 + c\alpha\beta/\pi}} + \frac{\beta}{2} e^{2 ||V||\beta} (\ell F(\beta) + G(\beta)) e^{-f(\beta) \ell}
    \\ \leq (\ell \widetilde{F}(\beta, V) + \widetilde{G}(\beta, V)) e^{-\widetilde{f}(\beta) \ell}.
\end{multline}
where we have defined functions $\widetilde{F}$, $\widetilde{G}$ and $\widetilde{f}$ upper-bounding the sum, all of which are computable from the constants involved in the previous expressions. Note that if we retrace all of the constants (in particular, the expressions involving $F(\beta)$ and $G(\beta)$), we can observe that $\widetilde{F}, \widetilde{G} \in \mathcal{O}(\exp(||V|| \beta)) = \mathcal{O}(\exp(||H_B|| \beta))$ and $\widetilde{f} \in \mathcal{O}(\beta^{-1})$. This completes the proof.

\section{The quantum Markov property, message-passing, and the circle product}
\label{appx:markov}

\noindent Here, we present a collection of technical results and elaboration on the Markov property and message-passing, which are invoked throughout the paper.

\subsection{Markov quantum graphical models}

\noindent The following theorem is the main tool in the proof of convergence for the original quantum message-passing algorithms, in Ref.~\cite{leifer2008quantum}.

\begin{theorem}[Decomposition Theorem, Ref.~\cite{hayden2004structure}]
\label{thm:decomp}
Let $\mathcal{H} = \mathcal{H}^{A} \otimes \mathcal{H}^{B} \otimes \mathcal{H}^{C}$ be a quantum system. A state $\rho_{ABC} \in \mathcal{H}$ satisfies $S(A : C | B) = 0$ if and only if there exists a decomposition of $\mathcal{H}^{B}$ as
\begin{equation}
\mathcal{H}^{B} = \displaystyle\bigoplus_{j} \mathcal{H}^{B_{L_j}} \otimes \mathcal{H}^{B_{R_j}}
\end{equation}
such that
\begin{equation}
\label{eq:decomp}
\rho_{ABC} = \bigoplus_{j} q_j \rho_{A B_{L_j}} \otimes \rho_{B_{R_j} C}
\end{equation}
where $\rho_{A B_{L_j}} \in \mathcal{H}^{A} \otimes \mathcal{H}^{B_{L_j}}$, $\rho_{B_{R_j} C} \in \mathcal{H}^{B_{R_j}} \otimes \mathcal{H}^{C}$, and $\sum_{j} q_j = 1$.
\end{theorem}

\noindent We use this result to prove a small technical lemma, which is used in Sec.~\ref{sec:background} to conclude that taking partial traces preserves the Markov property.

\begin{lemma}
  \label{lem:mar_tr}
  Suppose $(G, \rho_V)$ is a bifactor network with $G$ a tree, posessing the Markov property. Then, the partial trace over a leaf node of the graphical model preserves the Markov property.
\end{lemma}

\begin{proof}
  Suppose $w$ is a leaf node of $G$. Let $p(w)$ be the unique parent of $w$. Let us pick some $U \subset V - w$. Let $W = n(U) - U$ and $X = V - (n(U) \cup U)$, where $n(U)$ are the neighbours of $U$ in $V - w$.
  In the case that $p(w) \in X$, note that for the triple of sets $U$, $W$, and $X \cup \{w\}$, $W$ consists of all neighbours of $U$ in $V$ and $X \cup \{w\}$ consists of all
  non-neighbours, as $w$ only neighbours $p(w)$, and is hence a non-neighbour itself. We have, from Thm.~\ref{thm:decomp}:
  \begin{equation}
    \text{Tr}_{w} \left[ \rho_V \right] = \text{Tr}_w \left[ \displaystyle\sum_{j} q_{j} \sigma_{U W_{L_j}} \otimes \sigma_{W_{R_j} (X \cup \{w\})} \right] = \displaystyle\sum_{j} q_{j} \sigma_{U W_{L_j}} \otimes \sigma_{W_{R_j} X}
  \end{equation}
  It follows that in this case, we can write $\text{Tr}_w \left[ \rho_V \right]$ in the form of Eq.\eqref{eq:decomp}. In the case that $p(w) \in U$, we repeat the same procedure with $U \cup \{w\}$ in the place of $U$. If $p(w) \in W$,
  we can put $w$ in either $U$ or $X$, and repeat the same procedure as when $p(w) \in X$. Thus, Thm.~\ref{thm:decomp} implies that the model is Markov.
\end{proof}
\noindent The reason why this proof doesn't work if $w$ isn't a leaf node is because we could choose a triple $X, W, U$ where $w$ has neighbours in both $U$ and $X$, so we wouldn't be able to ``place" $w$ in either $U$ or $X$ without contradicting the conditions on these sets that we require for the model to be Markov.

\subsection{Calculation of circle products}

\noindent Computation of the circle product (Def.~\ref{def:circle_prod}) is the most important subrotuine in any protocol which executes the quantum message-passing algorithm. As a result, in order to understand how errors accumulate during belief-propagation, it is necessary not only to consider errors due to non-Markovity (in the sliding-window case), but also numerical errors which are introduced when manipulating large matrices. As it turns out, so long as errors in computation of effective Hamiltonians at each step of the message-passing algorithm are controlled, then error accumulates linearly. In particular, we have the following result:
\begin{lemma}
    Suppose $||H_A' - H_A|| \leq \epsilon_A$ and $||H_B' - H_B|| \leq \epsilon_B$. Then
    \begin{equation}
        \left| \left| \frac{1}{\mathcal{Z}_{AB}} \exp(H_A) \odot \exp(H_B) - \frac{1}{\mathcal{Z}_{A'B'}} \exp(H_A) \odot \exp(H_B') \right| \right| \leq 2(\epsilon_A + \epsilon_B)
    \end{equation}
\end{lemma}

\begin{proof}
    For some matrix function $A(s)$ defined for $s \in [0, 1]$ and some fixed $B$, define
    \begin{equation}
        F(s) = \frac{\exp(A(s) + B)}{\text{Tr} \left[ \exp(A(s) + B) \right]} = \frac{\exp(A(s) + B)}{\mathcal{Z}(s)}
    \end{equation}
    and note that
    \begin{equation}
    || F(1) - F(0) || = \left| \left| \displaystyle\int_{0}^{1} \frac{d F(s)}{ds} \ ds \right| \right| \leq \displaystyle\int_{0}^{1} \left| \left| \frac{d F(s)}{ds} \right| \right| ds
    \end{equation}
    where we can take the derivative via product rule
    \begin{equation}
    \label{eq:df}
        \frac{d F(s)}{ds} = \frac{1}{\mathcal{Z}(s)} \frac{d}{ds} \exp(A(s) + B) - \frac{\exp(A(s) + B)}{\mathcal{Z}(s)^2} \frac{d}{ds} \mathcal{Z}(s) = \frac{1}{\mathcal{Z}(s)} \left( \frac{d}{ds} \exp(A(s) + B) - F(s) \frac{d}{ds} \mathcal{Z}(s) \right)
    \end{equation}
    and the derivative of the exponential follows from Duhamel's formula:
    \begin{equation}
      \frac{d}{ds} \exp(A(s) + B) = \displaystyle\int_{0}^{1} e^{(1 - t) [A(s) + B]} A'(s) e^{t [A(s) + B]} dt
    \end{equation}
    It follows from the cyclic property of the trace that
    \begin{align}
        \frac{d}{ds} \mathcal{Z}(s) = \mathrm{Tr} \left[ \frac{d}{ds} \exp(A(s) + B) \right] &= \displaystyle\int_{0}^{1} \mathrm{Tr} \left[ e^{(1 - t) [A(s) + B]} A'(s) e^{t [A(s) + B]} \right] dt \\ &= \displaystyle\int_{0}^{1} \mathrm{Tr} \left[ e^{A(s) + B} A'(s) \right] dt = \mathrm{Tr} \left[ e^{A(s) + B} A'(s) \right]
    \end{align}
    which implies, from Eq.~\eqref{eq:df}, that
    \begin{align}
        \left| \left| \frac{d F(s)}{ds} \right| \right| &= \left| \left| \frac{1}{\mathcal{Z}(s)} \displaystyle\int_{0}^{1} e^{(1 - t) [A(s) + B]} A'(s) e^{t [A(s) + B]} dt - \frac{1}{\mathcal{Z}(s)} \text{Tr} \left[ e^{A(s) + B} A'(s) \right] F(s) \right| \right|
        \\ & \leq \left| \left| \frac{1}{\mathcal{Z}(s)} \displaystyle\int_{0}^{1} e^{(1 - t) [A(s) + B]} A'(s) e^{t [A(s) + B]} dt \right| \right| + \left| \left| \frac{1}{\mathcal{Z}(s)} \mathrm{Tr} \left[ e^{A(s) + B} A'(s) \right] F(s) \right| \right|
        \\ & \leq \displaystyle\int_{0}^{1} \left| \left| \left( \frac{e^{A(s) + B}}{\mathcal{Z}(s)} \right)^{1 - t} A'(s) \left( \frac{e^{A(s) + B}}{\mathcal{Z}(s)} \right)^{t} \right| \right| dt + \left| \mathrm{Tr} \left[ \frac{e^{A(s) + B}}{\mathcal{Z}(s)} A'(s) \right] \right|
        \\ & \leq \displaystyle\int_{0}^{1} \left| \left| \frac{e^{A(s) + B}}{\mathcal{Z}(s)} \right| \right|^{1 - t} ||A'(s)|| \left| \left| \frac{e^{A(s) + B}}{\mathcal{Z}(s)} \right| \right|^{t} dt + \left| \left| A'(s) \right| \right| \leq 2 ||A'(s)||
    \end{align}
    and, immediately, $||F(1) - F(0)|| \leq 2 ||A'(s)||$. Now, we can prove the desired result. Note by the triangle inequality that
    \begin{align}
        \left| \left| \frac{1}{\mathcal{Z}_{AB}} \exp(H_A + H_B) - \frac{1}{\mathcal{Z}_{A'B'}} \exp(H_A' + H_B') \right| \right| & \leq \left| \left| \frac{1}{\mathcal{Z}_{AB}} \exp(H_A + H_B) - \frac{1}{\mathcal{Z}_{A'B}} \exp(H_A' + H_B) \right| \right|
        \\ & + \left| \left| \frac{1}{\mathcal{Z}_{A'B}} \exp(H'_A + H_B) - \frac{1}{\mathcal{Z}_{A'B'}} \exp(H_A' + H_B') \right| \right|
    \end{align}
    It follows from setting $A(s) = (1 - s) H_A + s H_A'$ and $B = H_B$, so $A(0) = H_A$, $H(1) = H_A'$, and $A'(s) = H_A' - H_A$, we have
    \begin{equation}
        \left| \left| \frac{1}{\mathcal{Z}_{AB}} \exp(H_A + H_B) - \frac{1}{\mathcal{Z}_{A'B}} \exp(H_A' + H_B) \right| \right| \leq 2 ||H_A' - H_A|| \leq 2 \epsilon_A
    \end{equation}
    Identical reasoning shows that the latter difference is bounded by $2 \epsilon_B$, which implies that
    \begin{equation}
        \left| \left| \frac{1}{\mathcal{Z}_{AB}} \exp(H_A + H_B) - \frac{1}{\mathcal{Z}_{A'B'}} \exp(H_A' + H_B') \right| \right| \leq 2(\epsilon_A + \epsilon_B)
    \end{equation}
    and the proof is complete.
\end{proof}

\noindent When calculating the circle products via completely numerical means, it is necessary to compute matrix logarithms and exponentials. Often, algorithms which carry out these operations will have some amount of associated error, which is dependent on the condition number of the matrices in which they are operating. In light of this fact, it is also worthwhile to consider a lower bound on the minimum eigenvalue of density operators which are constructed via circle products. First, recall two important theorems.

\begin{theorem}[Golden-Thompson]
\label{thm:golden_thompson}
For Hermitian operators $A$ and $B$, $\text{Tr}[e^{A + B}] \leq \text{Tr}[e^{A} e^{B}]$.
\end{theorem}

\begin{theorem}[Weyl]
Let $\lambda_i(A)$ denote the $i$-th eigenvalue of operator $A$, in descending order. Let $\lambda_{\text{min}}(A)$ denote the minimum eigenvalue, and let $\lambda_{\text{max}}(A)$ denote the maximum. Given Hermitian operators $N$ and $R$, as well as $M = N + R$, then
\begin{equation}
    \lambda_{i}(N) + \lambda_{\text{min}}(R) \leq \lambda_{i}(M) \leq \lambda_{i}(N) + \lambda_{\text{max}}(R).
\end{equation}
\end{theorem}

\noindent This leads to the main result:

\begin{lemma}[Eigenvalue lower-bound for circle product]
If $A$ and $B$ are non-singular density operators, then the following inequality holds:
\begin{equation}
    \lambda_{\text{min}} \left( \frac{1}{\mathcal{Z}} (A \odot B) \right) \geq \frac{\lambda_{\text{min}}(A) \lambda_{\text{min}}(B)}{\lambda_{\text{max}}(A)},
\end{equation}
where $\mathcal{Z} = \text{Tr}[A \odot B]$.
\end{lemma}

\begin{proof}
Let $\lambda_{\text{min}}(M)$ denote the minimum eigenvalue of $M$. Given unit trace operators $A$ and $B$ with eigenvalues in $(0, 1]$, it is clear that $\lambda_{\text{min}}(\log(A)) = \log(\lambda_{\text{min}}(A))$, with the same holding true for $\log(B)$, as $\log$ is monotone increasing on this interval. From Weyl's inequality,
\begin{equation}
    \log(\lambda_{\text{min}}(A)) + \log(\lambda_{\text{min}}(B)) \leq \lambda_{\text{min}}(\log(A) + \log(B))
\end{equation}
which immediately implies that
\begin{align}
    \lambda_{\text{min}}(A) \lambda_{\text{min}}(B) &= \exp \left( \log(\lambda_{\text{min}}(A)) + \log(\lambda_{\text{min}}(B)) \right) \\
    & \leq \exp (\lambda_{\text{min}}(\log(A) + \log(B)))
    \\ & = \lambda_{\text{min}}(\exp(\log(A) + \log(B))) = \lambda_{\text{min}}(A \odot B)
\end{align}
More generally, we let $\lambda_i(M)$ denote the $i$-th eigenvalue of $M$, where we arrange the eigenvalues in ascending order. We have
\begin{equation}
    \text{Tr} \left[ A \odot B \right] \leq \text{Tr} \left[ AB \right] \leq \displaystyle\sum_{i = 1}^{N} \lambda_i(A) \lambda_i(B) \leq \lambda_{\text{max}}(A) \left( \displaystyle\sum_{i = 1}^{N} \lambda_i(B) \right) = \lambda_{\text{max}}(A)
\end{equation}
where the first inequality follows from Golden-Thompson, the second from the von Neumann trace inequality (and the fact that $A$ and $B$ are positive). The final inequality also follows from the positivity of the eigenvalues of $A$ and $B$. It follows immediately that
\begin{equation}
    \lambda_{\text{min}} \left( \frac{1}{\mathcal{Z}} (A \odot B) \right) = \frac{1}{\text{Tr}[A \odot B]} \lambda_{\text{min}}(A \odot B) \geq \frac{\lambda_{\text{min}}(A) \lambda_{\text{min}}(B)}{\lambda_{\text{max}}(A)}
\end{equation}
and the proof is complete.
\end{proof}

\section{Miscellaneous results}
\label{appx:misc}

\noindent This section contains a collection of minor auxiliary lemmas which are utilized in prior sections of this work.

\begin{lemma}[Commutator bound]
\label{lem:commutator_bound}
    For operators $A$ and $B$, $||[A, B^n]|| \leq n ||B||^{n - 1} ||[A, B]||$.
\end{lemma}
\begin{proof}
    This is clearly true for $n = 1$. Let us assume the case of $n$, for the case of $n + 1$ we have
    \begin{align}
        ||[A, B^{n + 1}]|| &= || A B^{n + 1} - B^{n + 1} A || = || AB B^{n} - B A B^{n} + B A B^{n} - B B^{n} A||
        \\ & \leq ||[A, B]|| ||B||^{n} + ||B|| ||[A, B^{n}]|| \leq ||[A, B]|| ||B||^{n} + n ||B||^{n} ||[A, B]||
        \\ & = (n + 1) ||B||^{n} ||[A, B]||
    \end{align}
    which proves the statement.
\end{proof}

\begin{lemma}[Telescoping bound]
\label{lem:telescoping}
   Given unitary operators $U$, $V$, and Hermitian operator $O$, and positive integer $k$,
   \begin{equation}
       ||V^k O (V^{\dagger})^{k} - (UV)^k O (V^{\dagger} U^{\dagger})^k || \leq \displaystyle\sum_{j = 1}^{k} ||V^{j} O (V^{\dagger})^{j} - U V^{j} O (V^{\dagger})^{j} U^{\dagger}||.
   \end{equation}
\end{lemma}
\begin{proof}
    Note that
    \begin{align}
   V^k O (V^{\dagger})^{k} - (UV)^k O (V^{\dagger} U^{\dagger})^k &= \displaystyle\sum_{j = 0}^{k - 1}  (UV)^{j} V^{k - j} O (V^{k - j})^{\dagger} (V^{\dagger} U^{\dagger})^{j} - (UV)^{j + 1} V^{k - j - 1} O (V^{k - j - 1})^{\dagger} (V^{\dagger} U^{\dagger})^{j + 1} \nonumber
   \\ & = \displaystyle\sum_{j = 0}^{k - 1} (UV)^{j} \left( V^{k - j} O (V^{\dagger})^{k - j} - U V^{k - j} O (V^{\dagger})^{k - j} U^{\dagger} \right) (V^{\dagger} U^{\dagger})^{j}
   \end{align}
   Thus,
   \begin{align}
       ||V^k O (V^{\dagger})^{k} - (UV)^k O (V^{\dagger} U^{\dagger})^k || &\leq \displaystyle\sum_{j = 0}^{k - 1} || (UV)^{j} \left( V^{k - j} O (V^{\dagger})^{k - j} - U V^{k - j} O (V^{\dagger})^{k - j} U^{\dagger} \right) (V^{\dagger} U^{\dagger})^{j} ||
       \\ & = \displaystyle\sum_{j = 0}^{k - 1} || V^{k - j} O V^{k - j} - U V^{k - j} O (V^{\dagger})^{k - j} U^{\dagger} ||
       \\ & = \displaystyle\sum_{j = 1}^{k} ||V^{j} O (V^{\dagger})^{j} - U V^{j} O (V^{\dagger})^{j} U^{\dagger}||
   \end{align}
   and the proof is complete.
\end{proof}
\begin{lemma}[Exponential bound]
\label{lem:exp_bnd}
    Given operators $A$ and $B$ and $M = \max(||A||, ||B||)$, $||\exp(A) - \exp(B)|| \leq e^{M} ||A - B||$.
\end{lemma}
\begin{proof}
    This follows from the fact that
    \begin{align}
        ||\exp(A) - \exp(B)|| & \leq \displaystyle\sum_{j = 1}^{\infty} \frac{1}{j!} ||A^j - B^j|| = \displaystyle\sum_{j = 1}^{\infty} \frac{1}{j!} \left|\left| \displaystyle\sum_{k = 1}^{j} A^{j - k} (A - B) B^{k - 1} \right|\right| \\ &\leq \displaystyle\sum_{j = 1}^{\infty} \frac{1}{j!} j M^{j - 1} ||A - B|| = e^{M} ||A - B||,
    \end{align}
    and we are done.
\end{proof}

\begin{lemma}[Partial trace decreases trace norm]
\label{lem:trace_norm}
Let $A$ be a Hermitian operator supported on $\mathcal{H} = \bigotimes_{v \in \mathcal{V}} \mathcal{H}^v$. Pick $\mathcal{V}' \subset \mathcal{V}$, then $||\text{Tr}_{\mathcal{V'}}\left[A\right] ||_1 \leq ||A||_1$.
\end{lemma}
\begin{proof}
Recall the variational definition of the trace norm:
\begin{equation}
    ||A||_1 = \sup_{||B|| \leq 1} | \text{Tr}(BA) |
\end{equation}
It follows immediately that
\begin{align}
    ||\text{Tr}_{\mathcal{V}'}\left[A\right]||_1 = \sup_{||B|| \leq 1} |\text{Tr}(B \text{Tr}_{\mathcal{V}'}\left[A\right])| &= \sup_{||B|| \leq 1} | \text{Tr}\left( (B_{\mathcal{V}'} \otimes \mathbb{I}_{\mathcal{V} - \mathcal{V}'}) A \right)|
    \\ & \leq \sup_{||C|| \leq 1} | \text{Tr}(C A) | = ||A||_1
\end{align}
where the supremum over $B$ is over matrices supported on $\mathcal{V}' \subset \mathcal{V}$, and the supremum over $C$ is over matrices supported on the entirety of $\mathcal{V}$.
\end{proof}

\end{document}